\documentclass[a4paper,oneside,reqno]{amsart}
\usepackage{amssymb,amsfonts}

\usepackage{hyperref}

\usepackage{float}
\usepackage[capitalise,noabbrev,sort]{cleveref}
\usepackage{thmtools}

\usepackage{cite}
\numberwithin{equation}{section} 
\usepackage[super]{nth} 

\usepackage{subcaption}
\usepackage{graphicx}
\usepackage{cancel}
\usepackage{tikz}
\usetikzlibrary{arrows,shapes,positioning,automata}
\usetikzlibrary{decorations.markings}
\tikzstyle arrowstyle=[scale=1]
\tikzstyle directed=[postaction={decorate,decoration={markings,
    mark=at position .5 with {\arrow[arrowstyle]{stealth}}}}]
\usepackage{amsmath, mathabx, tikz-cd}
\usetikzlibrary{decorations.pathmorphing}
\usetikzlibrary{arrows.meta}
\tikzcdset{arrow style=tikz,
    squigarrow/.style={
        decoration={
        snake, 
        amplitude=.4mm,
        segment length=2mm
        }, 
        rounded corners=.2pt,
        decorate
        }
    }

\renewcommand{\epsilon}{\varepsilon}

\usepackage{multicol}

\usepackage[a4paper]{geometry}

\usepackage[only,llbracket,rrbracket,llparenthesis,rrparenthesis]{stmaryrd} 
\usepackage{xcolor}

\usepackage{url,amscd,epsfig,enumerate,graphicx,amsfonts,latexsym}
\usepackage{comment}

\newcommand{\alephEND}{\aleph_{\sigma(1)}}
\newcommand{\APerm}{\mathcal A_\sigma}

\newcommand{\constPerm}{k+N+3}
 \newcommand{\bi}{\begin{itemize}}
\newcommand{\ei}{\end{itemize}} \newcommand{\be}{\begin{enumerate}}
\newcommand{\ee}{\end{enumerate}} \newcommand{\bc}{\begin{center}}
\newcommand{\ec}{\end{center}} \newcommand{\bt}{\begin{tabular}}
\newcommand{\et}{\end{tabular}} 
\newcommand{\ba}{\begin{array}} \newcommand{\ea}{\end{array}}

\newcommand{\Dendrite}{Final tree}\newcommand{\dendrite}{final tree}

\usepackage{mathrsfs}

\newcommand{\MCF}{multiple context-free}

\newcommand{\dom}{\operatorname{dom}}
\newcommand{\TSG}{\operatorname{TS}(\CCC)}
\newcommand{\up}{\operatorname{up}}
\newcommand{\down}{\operatorname{down}}\newcommand{\south}{\operatorname{south}}

\newcommand{\true}{\operatorname{true}}
\newcommand{\id}{\operatorname{id}}
\newcommand{\equals}{\operatorname{eq}}

\newcommand{\set}{\operatorname{set}}
\newcommand{\push}{\operatorname{push}}

\newcommand{\CCC}{C}

\newcommand{\cA}{\mathcal A}

\newcommand{\g}{\gamma }

\renewcommand{\d}{\delta }

\newcommand{\s}{\sigma }

\renewcommand{\geq}{\geqslant} \renewcommand{\leq}{\leqslant}  \renewcommand{\le}{\leqslant}

\newcommand{\SigHash}[1]{\Sigma[#1]}
\newcommand{\Cspec}[1]{C_{\text{spec}}}
\newcommand{\addRoot}{\prec_{\text{addRoot}}}

\usepackage[capitalise,noabbrev,sort]{cleveref}

\newcommand{\N}{\mathbb N}
\newcommand{\Z}{\mathbb Z}
\theoremstyle{plain}
\newtheorem{theorem}{Theorem}[section]
\newtheorem{lemma}[theorem]{Lemma}
\newtheorem{proposition}[theorem]{Proposition}

\theoremstyle{definition}
\newtheorem{definition}[theorem]{Definition}

\newtheorem{example}[theorem]{Example}
\newtheorem{remark}[theorem]{Remark}

\newtheorem{assumption}[theorem]{Assumption}
\newtheorem{propx}[theorem]{Property}

\newtheorem{theoremx}{Theorem}

\usepackage{multicol}
\renewcommand{\degree}{\deg}

\title[Permutation closure for multiple context-free languages]{Permutation closure for \\multiple context-free languages}

\author[A. Duncan]{Andrew Duncan}
\address{School of Mathematics, Statistics and Physics, Newcastle University,
Newcastle upon Tyne
NE1 7RU, United Kingdom}
\email{andrew.duncan@newcastle.ac.uk }

\author[M. Elder]{Murray Elder}
\address{School of Mathematical and Physical Sciences, University of Technology Sydney, Ultimo NSW 2007, Australia}
\email{murray.elder@uts.edu.au}

\author[L. Frenkel]{Lisa Frenkel}
\email{lizzy.frenkel@gmail.com}

\author[M. Lyu]{Mengfan Lyu}
\address{School of Computer, Data and Mathematical Sciences, Western Sydney University, NSW, Australia}
\email{mengfan.lyu@westernsydney.edu.au}

\date{\today}



\begin{document}

\begin{abstract}

 We prove  that the \emph{permutation closure} 
of a  multiple context-free language is   multiple context-free, which 
extends work of Okhotin and  Sorokin [LATA 2020] who showed closure under \emph{cyclic shift}, and 
complements work of  Brandst\"adt [1981, RAIRO Inform. Th\'{e}or.] (resp. 
 Brough \emph{et al.} [2016, Discrete Math. Theor. Comput. Sci.]) who showed  the same result for regular, context-sensitive, recursively enumerable (resp. EDT0L and ET0L) languages.
In contrast to Okhotin and  Sorokin who work with grammars, our proof uses restricted tree stack automata 
due to Denkinger [DLT 2016].

\end{abstract}

\maketitle

\section{Introduction}

Multiple context-free
  languages are a generalisation of context-free languages, introduced by 
   Seki,  Matsumura, Fujii and     Kasami        \cite{Kasami2, Seki}. Each \MCF\ language can be parameterised by a positive integer $k$, where a language is $k$-\MCF\ if it is generated by a $k$-\MCF\ grammar, or equivalently  it is accepted by a $k$-restricted tree-stack automaton  \cite{Denkinger}. Context-free languages are exactly the $1$-\MCF\ languages.

For a language $L$ and positive natural number $N$, 
Brandst\"{a}dt \cite{Brandst} considered the languages 
\[C^N(L)=\left\{w_{\sigma(1)}\cdots w_{\sigma(N)}\middle| w_1\cdots w_N\in L,\sigma\in S_N\right\}\] 
where  $S_N$ is the symmetric group on $N$ letters, 
and showed that regular, context-sensitive and recursively enumerable languages are closed under the operator $C^N$.
Brough, Ciobanu, Elder and Zetzsche 
  \cite{BroughPerm} showed that ET0L and EDT0L languages are also closed  under the operator $C^N$.
 It is clear that if the Parikh image of a language $L$ is semi-linear then so is that of the permutation closure of $L$, so having semi-linear Parikh image is closed under $C^N$.
 The operator $C^2$ is called  {\em cyclic shift}. Context-free languages are closed under cyclic shift
\cite{Maslov,Oshiba}, 
but  not closed under $C^3$ \cite{Brandst}. 

Recently Okhotin and  Sorokin \cite{cyclicShift}  showed that, for any positive integer $k$, the class of $k$-\MCF\ languages is  preserved under  cyclic shift.
Here we radically 
extend this result to show that being multiple context-free is closed under the operator $C^N$ for any $N$.

\begin{theoremx}\label{thm:Perm}Let $k, N\in\N_+$. The permutation closure $C^N(L)$ of a $k$-multiple context free language $L$ is  $(\constPerm)$-\MCF. 
\end{theoremx}

It follows, for example, that $C^N$ of a context-free language is $(N+4)$-\MCF\ (and ET0L by \cite{BroughPerm}).

In contrast to Okhotin and  Sorokin \cite{cyclicShift} who work with grammars, our proof relies on the $k$-restricted tree stack characterisation  \cite{Denkinger} due to Denkinger. 
The proof of our main lemma (\cref{lem:PermTech}) is technical, 
 but the main idea is that one can modify a restricted tree stack automaton by enlarging the state set and vertex label alphabet, 
 to record a finite amount of additional data in each vertex label of the tree stack which can be used to simulate reading factors of a word $w_1\cdots w_N$ in an alternative order (according to a fixed permutation $\sigma\in S_N$). 

Okhotin and  Sorokin also consider the class of \emph{well-nested} $k$-\MCF\ languages, and show that    the cyclic shift  of a {well-nested} $k$-\MCF\ language is a {well-nested} $(k+1)$-\MCF\ language.

Note that
Brandst\"{a}dt's definition of 
$C^N(L)$ 
 is for a fixed value $N$. 
Suppose instead we define  \[D(L)=\{w_{\sigma(1)}\ldots w_{\sigma(N)}\mid w_1\ldots w_N\in L,\sigma\in \cup_{N=1}^\infty S_N\}.\] Then even regular languages are not closed under the operator $D$, since for example 
 $D((abc)^*)= \text{MIX}_3$, which is  $2$-multiple context-free \cite{Salvati} and not context-free.

The article is organised as follows. In \cref{sec:Prelim} we furnish the definition of $k$-restricted tree stack automaton, fix notation and state some basic results, then in \cref{sec:Perm} we prove the main result.

\section{Preliminaries}\label{sec:Prelim}

\subsection{Notation} 

We use $\N_0,\N_+$ to denote the sets of natural numbers starting at $0,1$ respectively, and write $[i,j]$ for the set of integers
$\{k\in\Z\,\mid\, i\le k\le j\}$ when  $i,j\in\Z,i\leq j$.
If $\Sigma$ is a set, we let denote $\Sigma^*$ the set of all words (finite length strings) with letters from $\Sigma$, including the \emph{empty word} $\varepsilon$ of length $0$, $\Sigma^+$ the set of positive length words with letters from $\Sigma$, and $|\Sigma|$ the cardinality of the set. 
We use $\sqcup$ to denote the disjoint union of two sets. An \emph{alphabet} is a finite set.
For any set $\Sigma$ and letter $x\not\in\Sigma$ we let $\Sigma_x$ denote $\Sigma\sqcup\{x\}$. 
 If $u,v\in \Sigma^*$ we say $v$ is a \emph{factor} of $u$ if there exist $\alpha,\beta\in\Sigma^*$ so that $u=\alpha v\beta$.

For  the grammar definition of a $k$-multiple context free language see \cite{Seki,Salvati}.
We require the following well-known facts.
\begin{proposition}[See for example \cite{Seki}]\label{prop:closure_props}
  If $\Sigma,\Gamma$ are alphabets, $L,L_1\subseteq \Sigma^*$  are $k$-multiple context-free and $R\subseteq \Sigma^*$ is a regular language
  then 
\be[(i)]
\item \label{item:1} $L\cup L_1$ is $k$-multiple context-free;
\item\label{item:2} $L\cap R$ is $k$-multiple context-free; 
\item\label{item:3}  $\psi(L), \phi^{-1}(L)$ are $k$-multiple context-free for any homomorphisms $\psi\colon\Sigma^*\to \Gamma^*, \phi\colon\Gamma^*\to \Sigma^*$.
\ee\end{proposition}

\subsection{Tree stack automata}\label{subsec:TSA}

We recall the definitions of trees and tree stacks, following \cite{Substitution} (see also \cite{Denkinger,KS}).

\begin{definition}[Tree with labels in $\CCC$]
\label{defn:TreeWithLabels}
Let  $\CCC$ be an alphabet and $@\not\in\CCC$.
  Let $\xi$ be a  partial function from $\N_+^*$ to $\CCC_@$ ($=C\sqcup\{@\}$) with a non-empty, finite and prefix-closed domain denoted $\dom(\xi)$ such that $\xi(\nu)=@$ if and only if $\nu=\varepsilon$. 
  Then $\xi$ defines a labelled  rooted tree $T$ with vertices $V(T)=\dom(\xi)$,  where $\nu\in V(T)$ is labelled $\xi(\nu)$,  root $\varepsilon$,
  and edges $E(T)=\{\{\nu,\nu n\} \mid \nu n \in \dom(\xi)\}$. 
  \end{definition}

We will occasionally abuse notation and write $\xi=T$.
Note that $\dom(\xi)$ being non-empty and prefix-closed ensures that every tree with labels in $\CCC$ includes a root vertex $\epsilon$, and every vertex is connected to the root vertex $\epsilon$ by an edge path.

\begin{definition}[Below/above, child/parent] \label{defn:below-above}
In a rooted tree $T$ with vertex $\nu$, we say a vertex  is \emph{below} $\nu$ if it lies in the connected component of $T\setminus \{\nu\}$ containing the root, and \emph{above} $\nu$ if it is not below (so $\nu$ itself is considered to be \emph{above} $\nu$ by this definition).  Equivalently, $\mu$ is below $\nu$ if $\mu$ is not a prefix of $\nu$, and above $\nu$  if $\nu$ is a prefix of $\mu$. 
If $\nu n\in \dom(\xi)$ for some $n\in N_+$ then $\nu n$ is above $\nu$ by definition. As standard, we call $\nu n$ a
 \emph{child} of $\nu$, and say $\nu$ is the \emph{parent} of $\nu n$. \end{definition}

Note  that no vertex is below the root by this definition, and 
 that there is no requirement in the definition of a tree that a vertex  $\rho$ with a child $\rho n$, $n>1$, also has a child $\rho m$ for $m\in[1,n-1]$.

  \begin{definition}[Unlabelled tree]
  \label{defn:UnlabelledTree}
Let  $T$ be the labelled rooted tree defined by $\xi\colon \N_+^*\to\CCC_@$ as in \cref{defn:TreeWithLabels}.  Define the \emph{unlabelled  tree}  associated to $T$,
denoted by $\overline T$, as the tree with unlabelled vertices $V(\overline T)=\dom(\xi)$ and edges $E(\overline T)=E(T)$.
In this case we say $T$ \emph{has unlabelled tree} $\overline T$.
\end{definition}

\begin{definition}[Tree stack over $\CCC$]\label{def:treeStack} Let  $\xi\colon \N_+^*\to\CCC_@$ be a tree with labels,
and  $\rho\in \dom(\xi)$. The pair $(\xi,\rho)$ is called a 
 \emph{tree stack}.
The set of all tree stacks over $\CCC$ is denoted by $\TSG$.
\end{definition}
We refer to the address $\rho$ as the \emph{pointer} of the tree stack, and visualise the tree stack as a rooted tree with root at address $\varepsilon$ labeled $@$ and all other vertices labeled by letters from $\CCC$, with the {pointer} $\rho$ indicated by an arrow pointing to  the vertex at address $\rho$.

 \cref{fig:TreeStackEG} gives an example of \cref{def:treeStack} (figure (a)) and \cref{defn:UnlabelledTree} (figure (b)).

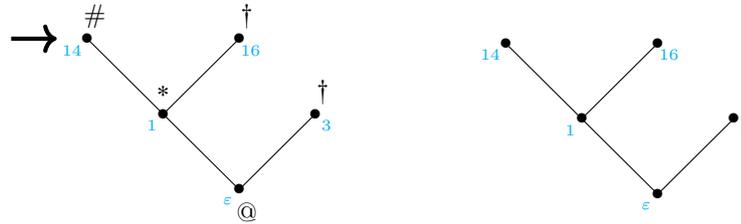
\begin{figure}[h!]   
\centering
\begin{subfigure}{.4\textwidth}
\begin{center}
\begin{tikzpicture}[scale=1]
\draw[decorate]  (0,2) --  (1,1) -- (2,0) -- (3,1);
\draw[decorate] (1,1) -- (2,2);

\draw (0,2) node {$\bullet$};
\draw (1,1) node {$\bullet$};
\draw (2,2) node {$\bullet$};
\draw (2,0) node {$\bullet$};
\draw (3,1) node {$\bullet$};

\draw (2.1,-0.3) node {$@$};
\draw (0.1,2.3) node {$\#$};
\draw (2.1,2.3) node {$\dag$};
\draw (3.1,1.3) node {$\dag$};
\draw (1,1.3) node {$\ast$};

\draw[decorate,ultra thick,->] (-1,2) -- (-0.4,2);
\draw (-0.2,1.85) node {\color{cyan}\tiny$14$};
\draw (2.15,1.85) node {\color{cyan}\tiny$16$};
\draw (0.85,0.85) node {\color{cyan}\tiny$1$};
\draw (3.15,0.85) node {\color{cyan}\tiny$3$};
\draw (1.85,-0.15) node {\color{cyan}\tiny$\varepsilon$};
\end{tikzpicture}\end{center}
\caption{Tree stack $(\xi, 14)$  with $T=\xi$  defined by  $\xi\colon \varepsilon\mapsto @, 1 \mapsto \ast, 14\mapsto \#, 16\mapsto \dag, 3\mapsto \dag$ and pointer at $12$. Here $\dom(\xi)=\{\varepsilon, 1,14, 16, 3\}$ and $\{\ast,  \#, \dag\}\subseteq C$. 
\label{fig:TreeStackEGleft}}
\end{subfigure}
\begin{subfigure}{0.05\textwidth} \phantom{.}
\end{subfigure}
\begin{subfigure}{0.4\textwidth}
\begin{tikzpicture}[scale=1]
\draw[decorate]  (0,2) --  (1,1) -- (2,0) -- (3,1);
\draw[decorate] (1,1) -- (2,2);

\draw (0,2) node {$\bullet$};
\draw (1,1) node {$\bullet$};
\draw (2,2) node {$\bullet$};
\draw (2,0) node {$\bullet$};
\draw (3,1) node {$\bullet$};

\draw (2.1,-0.3) node {\phantom{$@$}};

\draw (-0.2,1.85) node {\color{cyan}\tiny$14$};
\draw (2.15,1.85) node {\color{cyan}\tiny$16$};
\draw (0.85,0.85) node {\color{cyan}\tiny$1$};
\draw (3.15,0.85) node {\color{cyan}\tiny$3$};
\draw (1.85,-0.15) node {\color{cyan}\tiny$\varepsilon$};
\end{tikzpicture}
\caption{Unlabelled tree $\overline T$ associated to $T$. Note that the small digits drawn next to each vertex are addresses rather than labels.
\label{fig:TreeStackEGleft}}
\end{subfigure}

\caption{Tree stack and associated unlabelled tree, as in  \cref{def:treeStack,defn:UnlabelledTree}. Note the vertex $16$ is below the vertex $3$ according to \cref{defn:below-above}.
\label{fig:TreeStackEG}}
\end{figure}

To define a tree stack automaton, we 
define the following predicates and {(partial)} functions. 
The {predicates} are used to check if the pointer of a tree stack has a certain label, namely: 
\be \item $\equals(c)=\{(\xi, \rho) \in \TSG \mid \xi(\rho)=c\}$ where $c \in \CCC_@$ (so $(\xi, \rho)\in  \equals(c)$, or $\equals(c)$ is true for the tree stack  $(\xi, \rho)$, if the label of $\rho$ is $c$)
\item $\true$ 
 (for any tree stack regardless of the current label pointed to);

\ee

and the 
following (partial) {functions} enable 
 both construction of and movement around a tree: 
  \be \item $\id\colon \TSG \rightarrow \TSG$ where $\id(\xi, \rho)=(\xi, \rho)$ for every $(\xi, \rho) \in \TSG$. This is a function that makes no change to the tree stack
\item  For each $(n,c)\in \N_+\times \CCC$ let  $\push_n(c) \colon \TSG  \rightarrow \TSG$  be the
  partial function defined,  whenever $\rho n$ is not an address in $\dom(\xi)$, 
  by $\push_n(c)((\xi, \rho))=(\xi', \rho n)$, 
  where $\xi'$ is the partial function defined as 
\[\xi'(\nu)=\begin{cases}\xi(\nu), & \nu\in\dom(\xi)\\
  c, & \nu=\rho n\end{cases}.\]
Thus $\push_n(c)$  adds a new vertex $\rho n$ labelled by $c$, a new edge from $\rho$ to $\rho n$ and moves the pointer from $\rho$ to $\rho n$,  provided $\rho n$ was not already a vertex of $\xi$
\item   For each $n\in \N_+$ let $\up_n \colon \TSG \rightarrow \TSG$ be 
  the partial function defined,  whenever  $\rho n \in \dom(\xi)$,  by $\up_n((\xi, \rho))=(\xi, \rho n)$.
  So $\up_n$ moves the pointer of a tree stack to an existing child address
\item   $\down \colon \TSG \rightarrow \TSG$ is a partial function defined as $\down(\xi, \rho n)=(\xi, \rho)$ (so down  moves the pointer from a vertex to its parent and may be  applied to any tree stack unless its
  pointer is already pointing to the root, which has no parent
\item  For each $c\in \CCC$ let $\set(c) \colon \TSG  \rightarrow \TSG$  be the function defined,
  whenever $\rho\neq \varepsilon$, by  $\set(c)( (\xi,\rho))=(\xi',\rho)$, where $\xi'$ is the partial function  defined as 
\[\xi'(\nu)=\begin{cases}\xi(\nu), & \nu\in\dom(\xi)\setminus\{\rho\}\\
    c, & \nu=\rho\end{cases}.\]
Thus $\set(c)$ relabels the vertex pointed to by the pointer with the letter $c\in \CCC$ (and may be  applied to any tree stack except $(\xi,\varepsilon)$). 
\ee

Note that none of the partial functions $\id$, $\push_n(c)$, $\up_n$, $\down$ or $\set(c)$ remove any vertices from a tree stack, and only $\push_n(c)$ adds a vertex.

\begin{definition}[Tree stack automata]
A \emph{tree stack automaton (TSA)} is a tuple \[\mathcal{A}=(Q,\CCC, \Sigma, q_0,  \delta, Q_f)\] where $Q$ is a finite  \emph{state set}, $\CCC_@$ is an alphabet of \emph{tree-labels} (with $@ \notin \CCC$), $\Sigma$ is the alphabet of \emph{terminals}, $q_0\in Q$ is the \emph{initial state},  
 $Q_f\subseteq Q$ is the subset of \emph{final states}, and $\d = \{\s_1,\dots, \s_{|\delta|}\}$
  is a finite set of \emph{transition rules} of  the form $\s_i=(q,x,p,f,q')$ where 
 \be\item  $q,q'\in Q$ are the \emph{source} and \emph{target} states of $\s_i$,
 \item  $x\in\Sigma_\varepsilon$ is the \emph{input letter} for $\s_i$, 
 \item $p$ is the predicate of $\s_i$ (i.e. either $\equals(c)$, for some $c\in\CCC_@$, or $\true$),
 \item $f$ is the \emph{instruction} of $\s_i$, so is one of the functions $\id$, $\push_n(c)$, $\up_n$, $\down$, or $\set(c)$ for $c\in\CCC, n\in \N_+$.
 \ee
\end{definition}
Note that each coordinate of $\mathcal{A}$ has a finite description.
The tree stack automaton operates by starting in state $q_0$ with a tree stack initialised at $(\{(\varepsilon,@)\},\varepsilon)$,
reading each letter of an input word $w\in\Sigma^*$ with arbitrarily many $\varepsilon$ letters interspersed, applying transition rules from $\delta$, if they are allowed, (eg. $\push_n$ unless $\rho n$ is already a vertex;  the target state of the applied transition should also coincide with the source state of the following one). In other words, a composition of transitions is allowed if it determines
a well defined function which consumes input letters whilst moving between tree stacks. 

\begin{definition}[Run; accepted word; language of a TSA]
An allowed  sequence of transitions  $\tau_1\cdots \tau_r\in\delta^*$ is called a \emph{run} of the automaton. 
A run is \emph{valid} 
if it
finishes in a state in $Q_f$; if $w\in \Sigma^*$ is the input word consumed during a valid run we say that $\mathcal A$ \emph{accepts} 
$w$ (the definition does not \emph{a priori}
 put a requirement on the final position of the pointer, but see \cref{rmk:finish-root}).
Define the language $L(\mathcal A)$ to be the set of all words $w\in\Sigma^*$ accepted by $\mathcal A$.
 \end{definition}

\begin{example}\label{eg:ABCDMN} (See \cite[Example 3.3]{Denkinger}).
Let $Q=\{q_0, q_1, \ldots, q_9\}, \Sigma=\{a,b,c,d\}$ and $\CCC=\{*, \#\}$. Consider the TSA
$$
\mathcal{A}=(Q, \CCC, \Sigma, q_0, \delta, \{q_9\})
$$
where $\delta$ consists of the transitions
\begin{equation*}
\begin{array}{lll}
\s_1=(q_0, a, \equals(@), \push_1(*), q_1),     &\s_7=(q_3, \varepsilon, \true , \push_1(\#), q_4),    &\s_{13}=(q_6, \varepsilon, \equals(@), \up_2, q_7),\\
\s_2=(q_1, a , \true , \push_1(*), q_1),  &\s_8=(q_4, \varepsilon, \true , \down,  q_4),        &\s_{14}=(q_7, d, \equals(*) , \up_1, q_7),\\
\s_3=(q_1, \varepsilon, \true , \push_1(\#), q_2),  &\s_9=(q_4, \varepsilon,  \equals(@), \up_1, q_5),           &\s_{15}=(q_7, \varepsilon, \equals(\#) , \down, q_8),\\
\s_4=(q_2, \varepsilon, \true, \down, q_2), &\s_{10}=(q_5, c, \equals(*), \up_1, q_5),                                  &\s_{16}=(q_8, \varepsilon, \equals(*) , \down, q_8),\\
\s_5=(q_2, b, \equals(@) , \push_2(*), q_3), &\s_{11}=(q_5, \varepsilon, \equals(\#) , \down, q_6),                   &\s_{17}=(q_8, \varepsilon, \equals(@) , \id, q_9).\\
\s_6=(q_3, b, \true , \push_1(*), q_3),                  &\s_{12}=(q_6, \varepsilon, \equals(*) , \down, q_6),\\              
\end{array}
\end{equation*}

depicted by the labelled directed graph in \cref{fig:EGabcdAUTOM}.

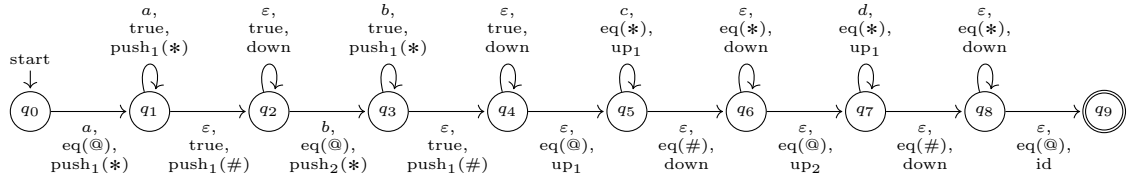
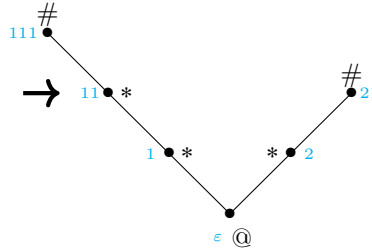
\begin{figure}[h]
\centering
\begin{subfigure}{1.0\textwidth}
\begin{center}
\begin{tikzpicture}[shorten >=1pt,node distance=1.57cm,on grid,auto] \tiny
 
 \node[state,initial, initial where=above, minimum size=0.5cm]  (1)                      {$q_0$};
  \node[state, minimum size=0.5cm]          (2) [right =of 1] {$q_1$};
  \node[state, minimum size=0.5cm]          (3) [right=of 2] {$q_2$};
    \node[state, minimum size=0.5cm]          (4) [right=of 3] {$q_3$};
    \node[state, minimum size=0.5cm]          (5) [right=of 4] {$q_4$};
    \node[state, minimum size=0.5cm]          (6) [right=of 5] {$q_5$};
    \node[state, minimum size=0.5cm]          (7) [right=of 6] {$q_6$};
    \node[state, minimum size=0.5cm]          (8) [right=of 7] {$q_7$};
    \node[state, minimum size=0.5cm]          (9) [right=of 8] {$q_8$};
  \node[state, minimum size=0.5cm, accepting](10) [right=of 9] {$q_9$};

  \path[->] (1) edge              node [below]      { $\begin{array}{c} a, \\ \equals(@) ,\\ \push_1(*)\end{array}$} (2)
                    
            (2) edge              node  [below]      {$\begin{array}{c}\varepsilon,\\ \true, \\ \push_1(\#)\end{array}$} (3)
                      
                           edge [loop above] node        {$\begin{array}{c}a,\\ \true ,\\\push_1(*) \end{array}$} ()
          
            (3) edge              node [below] {$ \begin{array}{c}b, \\ \equals(@), \\ \push_2(*)\end{array}$} (4)

                  edge [loop above] node    {$\begin{array}{c}\varepsilon,\\ \true, \\ \down \end{array}$} ()
                              
			(4) edge              node [below] {$ \begin{array}{c}\varepsilon, \\ \true, \\ \push_1(\#)\end{array}$} (5)

                  edge [loop above] node    {$\begin{array}{c} b,\\ \true, \\ \push_1(*) \end{array}$} ()
                     
            (5) edge              node  [below]      {$\begin{array}{c}\varepsilon,\\ \equals(@), \\ \up_1 \end{array}$} (6)
                      
                           edge [loop above] node        {$\begin{array}{c}\varepsilon,\\ \true ,\\ \down \end{array}$} ()
                
     		(6) edge              node  [below]      {$\begin{array}{c}\varepsilon,\\ \equals(\#), \\ \down \end{array}$} (7)
                      
                           edge [loop above] node        {$\begin{array}{c}c,\\ \equals(*) ,\\ \up_1 \end{array}$} ()
        
       		(7) edge              node  [below]      {$\begin{array}{c}\varepsilon,\\ \equals(@), \\ \up_2 \end{array}$} (8)
                      
                           edge [loop above] node        {$\begin{array}{c} \varepsilon,\\ \equals(*) ,\\ \down \end{array}$} ()
                                                
			(8) edge              node  [below]      {$\begin{array}{c}\varepsilon,\\ \equals(\#), \\ \down \end{array}$} (9)
                      
                           edge [loop above] node        {$\begin{array}{c} d,\\ \equals(*) ,\\ \up_1 \end{array}$} () 
                              
            (9) edge              node  [below] {$\begin{array}{c} \varepsilon, \\ \equals(@), \\\id \end{array}  $} (10) 
                  edge [loop above] node        {$\begin{array}{c} \varepsilon,\\ \equals(*), \\\down \end{array}$} () 
                    ;

\end{tikzpicture}
\end{center}
\caption{Finite state control encoding $\mathcal{A}$ in \cref{eg:ABCDMN}. 
 Each directed edge corresponds to a transition $\s_i\in\delta$.
\label{fig:EGabcdAUTOM}}

\end{subfigure}
\begin{subfigure}{0.45\textwidth}
\begin{center}
\begin{tikzpicture}[scale=0.8]
\draw[decorate]  (0,3) --  (1,2) -- (2,1) -- (3,0) -- (4,1) -- (5,2);

\draw (0,3) node {$\bullet$};
\draw (1,2) node {$\bullet$};
\draw (2,1) node {$\bullet$};
\draw (3,0) node {$\bullet$};
\draw (4,1) node {$\bullet$};
\draw (5,2) node {$\bullet$};

\draw (3.2,-0.4) node {$@$};
\draw (0.0,3.3) node {$\#$};
\draw (5.0,2.3) node {$\#$};
\draw (1.3,2.0) node {$\ast$};
\draw (2.3,1.0) node {$\ast$};
\draw (3.7,1.0) node {$\ast$};

\draw[decorate,ultra thick,->] (-0.4,2) -- (0.2,2);

\draw (-0.4,3.0) node {\color{cyan}\tiny$111$};
\draw (5.3,2.0) node {\color{cyan}\tiny$21$};
\draw (0.7,2.0) node {\color{cyan}\tiny$11$};
\draw (1.7,1.0) node {\color{cyan}\tiny$1$};
\draw (4.3,1.0) node {\color{cyan}\tiny$2$};
\draw (2.8,-0.4) node {\color{cyan}\tiny$\varepsilon$};
\end{tikzpicture}
\end{center}
\caption{Tree stack with pointer at $\rho=11$ after performing $\s_1 \s_2\s_3\s_4^3\s_5\s_7\s_8^2 \s_9 \s_{10}$ reading $a^2 b c$. \label{fig:EGabcd_FinalTree}}
\end{subfigure}

\caption{The finite state control $\mathcal{A}$ accepting $a^mb^nc^md^n$, $m,n \geq 0$; and a tree stack constructed whilst reading $a^2bc$ in \cref{eg:ABCDMN} (on the way to accepting $a^2bc^2d$).}
\label{fig:EGabcd}
\end{figure}

As shown in \cite{Denkinger}, $L(\mathcal{A}) = \{a^nb^mc^nd^m|m,n\in \N_+\}$. For instance, one may verify that  the TSA $\mathcal{A}$ accepts $a^2bc^2d$ by using the 
 run 
   \[\mathcal R= \s_1 \s_2 \s_3 \s_4^3 \s_5 \s_7 \s_8^2 \s_9 \s_{10}^2 \s_{11} \s_{12}^2 \s_{13} \s_{14} \s_{15} \s_{16} \s_{17}.\] 
\end{example}

\begin{definition}[\Dendrite]
Let $\mathcal A$ be a TSA, $\mathcal R$ a run of $\mathcal A$ accepting a word $w$, and  $(T_0,\varepsilon)$ the tree stack with pointer at the root after $\mathcal R$ has been performed. Call $T_0$ the \emph{\dendrite} corresponding to $\mathcal R$.
\end{definition}
For example, the \dendrite\ corresponding to the run  \begin{equation}\label{eq:runEG}
\mathcal R= \s_1 \s_2^{p-1} \s_3 \s_4^{p+1} \s_5 \s_6^{q-1} \s_7 \s_8^{q+1} \s_9 \s_{10}^p \s_{11} \s_{12}^p \s_{13} \s_{14}^q \s_{15} \s_{16}^q \s_{17}\end{equation} accepting $a^pb^qc^pd^q$ for $p,q\in\N_+$ 
in \cref{eg:ABCDMN} will be two path-graphs (of length $p+1$ and $q+1$) joined at a root vertex,  such as that drawn for $p=2,q=1$ in \cref{fig:EGabcd_FinalTree}.

The next restriction will put a bound on the number of times the pointer can move a step in the direction away from the root (by a $\push_n$ or $\up_n$ instruction) at a given vertex.

\begin{definition}[visited from below]
Let $\nu'\in \N_+^*$, $n\in\N_+$, $\nu=\nu' n$, and suppose that some run $\tau_1 \cdots \tau_r$  of a TSA  $\mathcal A$ 
has \dendrite\ containing a vertex $\nu$. We say that $\nu$ is 
\emph{visited from below} by $\tau_j$
if $\tau_j=(q,x,p,\push_n(c),q')$, $c\in\CCC$ or  $\tau_j=(q,x,p,\operatorname{up}_n,q')$ is applied when  the pointer is at $\nu'$.
\end{definition}
\begin{definition}[$k$-restricted]
 We say the run  $\tau_1 \cdots \tau_r$ is \emph{$k$-restricted} if for each $\nu\in \N_+^+$ the number of 
$j\in[1,r]$ such  that 
$\nu$ is visited from below by $\tau_j$ is at most $k$. We say that 
a TSA $\mathcal A$ is \emph{$k$-restricted} if for any word $w\in L(\mathcal A)$
there is an accepting run that is $k$-restricted.\end{definition}
For example, one may verify that the automaton in \cref{eg:ABCDMN} accepting $\{a^mb^nc^md^n|m,n \in\N_0\}$ is 2-restricted as follows: the word $a^p b^q c^p d^q$ is accepted by the  run given by \cref{eq:runEG}.
For this run, each vertex address $1^i$ for $i\in[1,p+1]$  is visited from below once with a $\s_2$ or $\s_3$ 
($\push_1$ transition), once with a $\s_9$ or $\s_{10}$ ($\up_1$ transition), while each vertex address $21^j$ for $j=0$  is visited from below once with $\s_5$ ($\push_2$ transition) and once with $\s_{13}$ ($\up_2$ transition), and for $j\in[1,q]$  is visited from below once with a $\s_6$ or $\s_7$ ($\push_1$ transition) and another time with $\s_{14}$ ($\up_1$ transition).  
 The language  $\{a^mb^nc^md^n\mid n\in\N_0\}$
 can also be generated by a $2$-\MCF\ grammar (see for example \cite[Example 3.3]{Denkinger}).

We now state the characterisation of $k$-\MCF\ languages 
 due to  Denkinger.
\begin{theorem}[Theorem 4.12, \cite{Denkinger}]$\label{equ}$
Fix an alphabet $\Sigma$. Let $L \subseteq \Sigma^*$ and $k \in \N_{+}$. The following are equivalent:
\be\item $L$ is $k$-\MCF\ (\emph{i.e.} there is a $k$-multiple context-free grammar ${G}$ which generates  $L$)
\item 
there is a $k$-restricted tree stack automaton $\mathcal{A}$ which recognises  $L$.\ee
\end{theorem}

A non-empty run of the TSA $\mathcal A$ has the form $\tau_1 \cdots \tau_r$, for some transitions $\tau_j \in \delta, j \in [1,r].$
Denote by $(\xi_j,\rho_j)$ the tree stack after performing $\tau_1\cdots \tau_j$,  
and $(\xi_0,\rho_0)$  the tree stack before any transitions are performed ($\dom(\xi_0)=\{\epsilon\}$, $\xi_0(\epsilon)=@$ and $\rho_0=\epsilon$).

\begin{definition}[Degree]
\label{defn:degree}
If $\delta$ is the set of  transition rules for a TSA $\mathcal A$, let $\Delta=\{i\in \N_+ \mid \exists c\in C, \push_i(c)\in \delta\}$. This means that in any tree stack $(\xi,\rho)$ built  by $\mathcal A$ whilst reading any input, $\dom(\xi)\subseteq \Delta^*$, so the out-degree of  any vertex of $\xi$ is at most
$|\Delta|$.
 Define the 
 \emph{degree} of $\mathcal A$ to be $\degree(\mathcal A)=|\Delta|$.
\end{definition}
Since $\delta$ is finite, $\degree(\mathcal A)$ is a non-negative integer. 
Recall (\cref{defn:TreeWithLabels})
that we do not require  vertices $\nu n$ in the final tree of a run to have a sibling $\nu (n-1)$, so it is possible that $\Delta$ contains integers larger than  $\degree(\mathcal A)$. We make the following assumption.

  \begin{assumption}
  \label{assumption:DegreeSharp}
 Without loss of generality 
 we may assume \[\max\{i\in \N_+ \mid \exists c\in C, \push_i(c)\in \delta\}=\degree(\mathcal A).\] \end{assumption}
If \cref{assumption:DegreeSharp} does not hold, there exists $j\geq 2$ so that $\mathcal A$ has an instruction $\push_{j}(c)$ for some $c\in\CCC$ but no instruction $\push_{j-1}(c')$ for any $c'\in\CCC$. 
Then replacing all rules $\push_j,\up_j$ in $\delta$ 
by $\push_{j-1},\up_{j-1}$ neither changes the language accepted nor affects the number of times a vertex is visited from below. We can then repeat until \cref{assumption:DegreeSharp}  is satisfied.

We make the following additional convention
 for tree stack automata appearing in this paper.

  \begin{assumption}[Pointer returns to the root]\label{rmk:finish-root}
 Without loss of generality (modifying the 
 TSA description by adding from all accept states additional $\down$ transitions to a new single accept state  if necessary) we may assume $\mathcal A$ accepts inputs only if it finishes in a state in $Q_f$ with the pointer pointing to the root. 
Note that adding $\down$ transitions does not affect the number of times a vertex is visited from below by a run.
 \end{assumption}

For example, 
the $2$-TSA given in \cref{eg:ABCDMN} accepts only when the pointer is at the root,  and satisfies \cref{assumption:DegreeSharp} with $\deg(\mathcal A)=2$.

\section{Proof of \cref{thm:Perm}}
\label{sec:Perm}

\cref{thm:Perm} follows from the next two lemmas 
and \cref{prop:closure_props} \eqref{item:1}.
For  $\Sigma$ an alphabet and  $s\in\N$,  let $\SigHash{s}=\Sigma\sqcup\{\#_1,\dots, \#_{s+1}\}$ where each letter $\#_i$ is distinct and not in $\Sigma$.

\begin{lemma}\label{lem:InsertHashes} Let $k,N\in \N_+$.
If $L\subseteq \Sigma^*$ is $k$-\MCF, then so is
 \[L_N=\left\{
\#_1w_{1}\#_2w_{2}\#_3\cdots\#_N w_N\#_{N+1}\middle| w_1w_2\cdots w_N\in L\right\}\subseteq (\SigHash{N})^*.
\]

Moreover, there is a $k$-restricted TSA  $\mathcal A =(Q,\CCC, \Sigma, q_0,  \delta, \{q_f\})$    with $\deg(\mathcal A)=D$, 
 such that  $q_f\neq q_0$, $L_N=L(\mathcal A)$, and for all $q'\in Q\setminus \{q_0,q_f\}$, 
\be\item 
$\delta$ has exactly one transition from $q_0$ to $q'$, which is $(q_0,\#_1, \equals(@),\id,q')$, 
and  exactly one transition  from $q'$ to $q_f$, which is  $(q',\#_{N+1}, \equals(@),\id,q_f)$,
   \ee
and 
 a $k$-restricted TSA $\mathcal A_N = (Q[N],\CCC[N], \SigHash{N}, q_0,  \delta_N, Q_f)$ with  $\deg(\mathcal A_N)=D+N+1$,
where 
\[C[N]=C\sqcup\{\llbracket i,q,q',c\rrbracket \mid  i\in[1,N+1],  q,q'\in Q,c\in C_@\}, 
\] 
   \[Q[N]=Q\sqcup\left\{\llbracket i,q,q',c\rrbracket_j\middle| i\in[1,N+1], j\in\{1,2\}, q,q'\in Q,c\in C_@\right\},\] 
such that  \be
\setcounter{enumi}{1}
\item $L_N=L(\mathcal A_N)$
\item   for all $x\in \Sigma_\varepsilon$, $\tau = (q,x, p,f,q')\in \delta$ if and only if $\tau\in \delta_N$
  \item for $i\in[1,N+1]$, and  for all $c\in C_@$
$\tau = (q,\#_i, p,f,q')\in\delta$ if and only if $\tau\not\in\delta_N$,   $f\in\{\push_n(c),\set(c),\up_n,\down,\id\}$, and 
    \begin{equation}\label{it:t1} 
    \begin{split} \tau_1 (i,c)&= (q,\varepsilon,  p,f,  \llbracket i,q,q',c\rrbracket_1)\in\delta_N\\
  \tau_2 (i,c)& = (\llbracket i,q,q',c\rrbracket_1,\epsilon, \equals(c),\push_{D+i}( \llbracket i,q,q',c\rrbracket)\llbracket i,q,q',c\rrbracket_2)\in\delta_N\\
    \tau_3 (i,c) &=(\llbracket i,q,q',c\rrbracket_2,\#_i,\equals(\llbracket i,q,q',c\rrbracket), \down, q')\in\delta_N. \end{split}\end{equation}
   \ee
  \end{lemma}

 \begin{proof}
 Let $\psi\colon(\SigHash{N})^*\to\Sigma^*$ be the homomorphism  induced by  the
   map from $\SigHash{N}$ to $\Sigma$ which restricts to the identity on $\Sigma$ and  for each $i\in[2,N]$ sends $\#_i$ to $\varepsilon$. Then let $L_{\text{temp}}$ be the intersection
   of $\psi^{-1}(L)$ with the regular language $\Sigma^*\#_2\Sigma^*\#_3\cdots\#_N \Sigma^*$, so by \cref{prop:closure_props} \eqref{item:2}  and
   \eqref{item:3}, $L_{\text{temp}}$ is $k$-\MCF.

  Let $\mathcal A_{\text{temp}}=(Q_{\text{temp}},\CCC, \SigHash{N}, (q_0)_{\text{temp}},  \delta_{\text{temp}}, (Q_f)_{\text{temp}})$ be 
   a $k$-restricted TSA accepting $L_{\text{temp}}$.  Recall 
   (\cref{rmk:finish-root}) that $\mathcal A_{\text{temp}}$ accepts only when the pointer is at the root.
Setting $Q=Q_{\text{temp}}\sqcup\{q_0,q_f\}$ where $q_0,q_f$ are distinct and not in $Q_{\text{temp}}$,  $Q_f=\{q_f\}$, and adding transitions 
\be\item $(q_0, \#_1, \equals(@), \id, (q_0)_{\text{temp}})$
\item$(q',\#_{N+1}, \equals(@),\id,q_f)$ for every $q'\in (Q_f)_{\text{temp}}$
 \ee yields a $k$-restricted TSA  $\mathcal A=(Q,\CCC, \SigHash{N}, q_0,  \delta,\{q_f\})$ accepting $L_N$ and satisfies the properties 
 \be
   \item for all $q'\in Q\setminus \{q_0,q_f\}$, $\delta$ has exactly one transition of the form $(q_0,\#_1, \equals(@),\id,q')$ 
\item for all $q'\in Q\setminus \{q_0,q_f\}$,     $\delta$ has exactly one transition of the form $(q',\#_{N+1}, \equals(@),\id,q_f)$.
   \ee given in the  statement of the lemma.

Construct a new $k$-restricted TSA $\mathcal A_N$ with state set $Q[N]$ and tree label alphabet $C[N]$ as in the statement of the lemma, and $\delta_N$ obtained from $\delta$ 
by  keeping all transitions in $\delta$ that read $x\in \Sigma_\varepsilon$,
 and replacing each  transition of the form  $(q,\#_i,p,f, q')\in \delta$ 
  by the (at most) $3|C|$ transitions of the form in \cref{it:t1} in the statement of the lemma.
   Note that these transitions can be followed if and only if the vertex reached after performing the transition $\tau=(q,\#_i,p,f, q')$ of $\mathcal A$ has the label $c$; if not, the second transition cannot be followed since $\equals(c)$ does not hold.

The effect of this change is as follows. If the transition $\tau=(q,\#_i,p,f, q')$ of $\mathcal A$
  leaves the pointer at a vertex $\nu$ with  label $c$ (where $f$ may be any of $\push_n(c),\set(c),\up_n$ for the case $c\neq @$, or $\down,\id$ for $c\in C_@$), then $\mathcal A_N$ will instead perform $(q,\varepsilon,  p,f,  \llbracket i,q,q',c\rrbracket_1)$  reading no input and moving the pointer to the same vertex $\nu$ labeled $c$, and move into the special state $ \llbracket i,q,q',c\rrbracket_1$. It then performs a second transition which    checks  the  label at $\nu$ is $c$, creates new vertex at address $\nu(D+i)$ with address labelled $ \llbracket i,q,q',c\rrbracket$, moves the pointer to this position, and moves the TSA into the state $ \llbracket i,q,q',c\rrbracket_2$. The third transition reads the letter $\#_i$, checks the agreement of labels, and moves down to leave the pointer pointing to $\nu$ with label $c$ and the TSA in state $q'$. 
  For an illustration of this process see \cref{fig:Eg-AddSpecial}.

  \begin{figure}[h]

\begin{subfigure}{1\textwidth}
\centering
\begin{tabular}{|c|c|}
\hline
 \begin{tikzpicture}[scale=.9]
\draw[decorate]  (0,2) --  (1,1);
\draw[decorate]   (1,1) -- (0,0);
\draw[decorate] (1,1) -- (2,2);
\draw[decorate] (-0.5,3)-- (0,2) -- (0.5,3);
\draw (0,2) node {$\bullet$};
\draw (1,1) node {$\bullet$};
\draw (2,2) node {$\bullet$};
\draw (0,0) node {$\bullet$};
\draw (-0.5,3) node {$\bullet$};
\draw (0.5,3) node {$\bullet$};
\draw (0.1,-0.3) node {$@$};
\draw (2.3,2) node {$b$};
\draw[decorate,ultra thick,->] (1.1,2) -- (1.6,2);
\draw (2.05,1.8) node {\color{cyan}\tiny$\nu$};

\phantom{           \draw (2,-.2) node [rotate=90] {\tiny $\llbracket 1,q_0,q',@\rrbracket$};}

\end{tikzpicture}
& \begin{tikzpicture}[scale=.9]
\draw[decorate]  (0,2) --  (1,1);
\draw[decorate]   (1,1) -- (0,0);
\draw[decorate] (1,1) -- (2,2);
\draw[decorate] (-0.5,3)-- (0,2) -- (0.5,3);
\draw (0,2) node {$\bullet$};
\draw (1,1) node {$\bullet$};
\draw (2,2) node {$\bullet$};
\draw (0,0) node {$\bullet$};
\draw (-0.5,3) node {$\bullet$};
\draw (0.5,3) node {$\bullet$};
\draw (0.1,-0.3) node {$@$};
\draw (2.3,2) node {$c$};
\draw[decorate,ultra thick,->] (1.1,2) -- (1.6,2);
\draw (2.05,1.8) node {\color{cyan}\tiny$\nu$};
\phantom{           \draw (2,-.2) node [rotate=90] {\tiny $\llbracket 1,q_0,q',@\rrbracket$};}
\end{tikzpicture}
\\
state: $q$ & $q'$\\
\hline
\end{tabular}

\caption{Single transition $\tau=(q,\#_i,\true,\set(c),q')\in\delta$ reading $\#_i$ for $i\in [2,N]$. 
\label{fig:lemmaAdd3Trans1}}

\end{subfigure}

\vspace{5mm}

\begin{subfigure}{1\textwidth}
\centering

\begin{tabular}{|c|c|}

\hline
 \begin{tikzpicture}[scale=.9]
\draw[decorate]  (0,2) --  (1,1);
\draw[decorate]   (1,1) -- (0,0);
\draw[decorate] (1,1) -- (2,2);
\draw[decorate] (-0.5,3)-- (0,2) -- (0.5,3);
\draw (0,2) node {$\bullet$};
\draw (1,1) node {$\bullet$};
\draw (2,2) node {$\bullet$};
\draw (0,0) node {$\bullet$};
\draw (-0.5,3) node {$\bullet$};
\draw (0.5,3) node {$\bullet$};
\draw (0.1,-0.3) node {$@$};
\draw (2.3,2) node {$b$};
\draw[decorate,ultra thick,->] (1.1,2) -- (1.6,2);
\draw (2.05,1.8) node {\color{cyan}\tiny$\nu$};

            \draw  (2,.8) node {$\bullet$};
              \draw   (0,0) --  (2,0.8);
            \draw (2.5,.5) node{\tiny $\llbracket 1,q_0,q',@\rrbracket$};
    
\end{tikzpicture}

&

 \begin{tikzpicture}[scale=.9]
\draw[decorate]  (0,2) --  (1,1);
\draw[decorate]   (1,1) -- (0,0);
\draw[decorate] (1,1) -- (2,2);
\draw[decorate] (-0.5,3)-- (0,2) -- (0.5,3);
\draw (0,2) node {$\bullet$};
\draw (1,1) node {$\bullet$};
\draw (2,2) node {$\bullet$};
\draw (0,0) node {$\bullet$};
\draw (-0.5,3) node {$\bullet$};
\draw (0.5,3) node {$\bullet$};
\draw (0.1,-0.3) node {$@$};
\draw (2.3,2) node {$c$};
\draw[decorate,ultra thick,->] (1.1,2) -- (1.6,2);
\draw (2.05,1.8) node {\color{cyan}\tiny$\nu$};

            \draw  (2,.8) node {$\bullet$};
              \draw   (0,0) --  (2,0.8);
             \draw (2.5,.5) node{\tiny $\llbracket 1,q_0,q',@\rrbracket$};
\end{tikzpicture}

\\&
\\
state: $q$
& $\llbracket i,q,q',c\rrbracket_1$ \\
\hline
 \begin{tikzpicture}[scale=.9]
 \draw[thick,magenta]  (2,2) --  (2,3.2);
\draw[decorate]  (0,2) --  (1,1);
\draw[decorate]   (1,1) -- (0,0);
\draw[decorate] (1,1) -- (2,2);
\draw[decorate] (-0.5,3)-- (0,2) -- (0.5,3);
\draw (0,2) node {$\bullet$};
\draw (1,1) node {$\bullet$};
\draw (2,2) node {$\bullet$};
\draw (0,0) node {$\bullet$};
\draw (-0.5,3) node {$\bullet$};
\draw (0.5,3) node {$\bullet$};
\draw (0.1,-0.3) node {$@$};
\draw (2.3,2) node {$c$};
\draw[decorate,ultra thick,->] (1.1,3.2) -- (1.6,3.2);
\draw (2.05,1.8) node {\color{cyan}\tiny$\nu$};
\draw (2.65,3) node {\color{cyan}\tiny$\nu(D+i)$};
\draw (2,3.2) node {\color{magenta}$\bullet$};
\draw (2.2,3.6) node { \footnotesize \tiny$\llbracket i,q,q',c\rrbracket$};

            \draw  (2,.8) node {$\bullet$};
              \draw   (0,0) --  (2,0.8);
             \draw (2.5,.5) node{\tiny $\llbracket 1,q_0,q',@\rrbracket$};
\end{tikzpicture}
&
 \begin{tikzpicture}[scale=.9]
 \draw[thick,magenta]  (2,2) --  (2,3.2);
\draw[decorate]  (0,2) --  (1,1);
\draw[decorate]   (1,1) -- (0,0);
\draw[decorate] (1,1) -- (2,2);
\draw[decorate] (-0.5,3)-- (0,2) -- (0.5,3);
\draw (0,2) node {$\bullet$};
\draw (1,1) node {$\bullet$};
\draw (2,2) node {$\bullet$};
\draw (0,0) node {$\bullet$};
\draw (-0.5,3) node {$\bullet$};
\draw (0.5,3) node {$\bullet$};
\draw (0.1,-0.3) node {$@$};
\draw (2.3,2) node {$c$};
\draw[decorate,ultra thick,->] (1.1,2) -- (1.6,2);
\draw (2.05,1.8) node {\color{cyan}\tiny$\nu$};
\draw (2.65,3) node {\color{cyan}\tiny$\nu(D+i)$};

\draw (2,3.2) node  {\color{magenta}$\bullet$};
\draw (2.2,3.6) node { \footnotesize\tiny $\llbracket i,q,q',c\rrbracket$};

            \draw  (2,.8) node {$\bullet$};
              \draw   (0,0) --  (2,0.8);
               \draw (2.5,.5) node{\tiny $\llbracket 1,q_0,q',@\rrbracket$};
\end{tikzpicture}
\\
&\\
state:

  $\llbracket i,q,q',c\rrbracket_2$ & $q'$\\

\hline
\end{tabular}

\captionsetup{justification=centering}
\caption{
Replacing $\tau$ by 
three transitions 
$\tau_1(i,c)= (q,\varepsilon,  \true,\set(c),  \llbracket i,q,q',c\rrbracket_1)$, \\
$\tau_2(i,c) = (\llbracket i,q,q',c\rrbracket_1,\varepsilon, \equals(c),\push_{D+i}( \llbracket i,q,q',c\rrbracket), \llbracket i,q,q',c\rrbracket_2)$, \\
$ \tau_3(i,c) =(\llbracket i,q,q',c\rrbracket_2,\#_i, \equals(\llbracket i,q,q',c\rrbracket), \down, q')
\in\delta_N$ 
\label{fig:lemmaAdd3Trans2}
}

\end{subfigure}

\caption{Transitions in the TSA $\mathcal A_N$ versus in $\cA$ when  $\tau=(q,\#_i,\true,\set(c),q')\in\delta$ 
in Lemma~\ref{lem:InsertHashes} and $i\in [2,N]$. For this illustration we are assuming we are replacing the transition which reads $\#_j$ in order starting from $j=1$ to $j=i$, so the root has a vertex at address $1(D+1)$ labeled $\llbracket 1,q_0,q',@\rrbracket$ and but not yet a vertex at address $1(D+N+1)$.
\label{fig:Eg-AddSpecial}}
\end{figure}
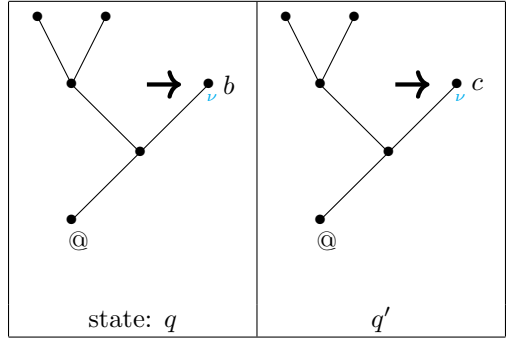
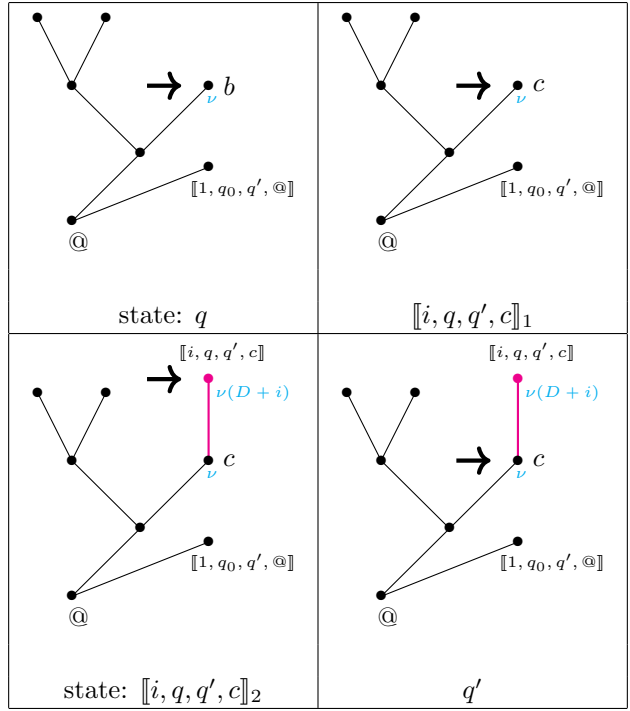

  It follows that $L(\mathcal A_N)=L_N$, and the above conditions are satisfied,  since this is the only change.  
 \end{proof}

\begin{definition}[Special vertices]
We call the vertices added to construct the tree  
 in \cref{lem:InsertHashes} 
 \emph{special}.  A special vertex with address $\nu(D+i)$  labeled 
$ \llbracket i,q,q',c\rrbracket$ is called an \emph{$i$-th special vertex}.
\end{definition} 

 \begin{remark}By construction, if $T_0$ is the \dendrite\ corresponding to a run $\mathcal R$
 of $\mathcal A_N$  accepting a word $w\in L_N$,  then 
    $T_0$ will have one $i$-th special vertex for each $i\in[1,N+1]$. Note that it is possible to have more than one special vertex connected by an edge to the same vertex $\nu$ 
    (for example, the root is connected to the \nth{1} and $(N+1)$-th special vertices, and possibly more).

The conditions in \cref{lem:InsertHashes} imply that if $\mathcal R\in  \delta^*$ is a run  accepting  a word $w\in L_N$, there is a  run  $\mathcal R'\in\delta_N^*$ accepting $w$ obtained from $\mathcal R$ by replacing the $N+1$ transitions which read the letters $\#_i$ by $3(N+1)$ of  the (at most) $3(N+1)|C|$ new transitions with appropriate choices for $q_j,q_j',c_j$ as in the lemma.
By construction, special vertices 
are visited by $\mathcal A_N$ from below exactly once in any run accepting a word $w\in L_N$.

 \end{remark}

  \begin{definition}[$H_{N,C}$, $H_{N,@}$, $\circledcirc_{N,D}$]\label{defn:setsHCirc} 
 Let  $N,D\in \N_+$ and $C$ be an alphabet with $@\not\in C$.
   Define $H_{N,C}$ to be the set of all  
$2\times N$ arrays of the form
\begin{equation}\label{eqn:modHA} 
\begin{pmatrix}
\% & \cdots & \%  & n_i & \cdots &   n_N \\
\% & \cdots & \%  & c_i &  \cdots &  c_N \\
\end{pmatrix}
\end{equation} where $i\in[1,N]$ and $n_{j} \in \{0,1,2\}$, $c_j\in C$  for each $j\in[i,N]$,
 $H_{N,@}$ to be the set of all 
$2\times N$ arrays of the form
\begin{equation}\label{eqn:modHA} 
\begin{pmatrix}
n_1  & n_2  &  \cdots &   n_N  \\
 @ &       @ & \cdots &  @ \\
\end{pmatrix}
\end{equation} where $n_{i} \in \{0,1,2\}$  for each $i\in[2,N]$, $n_{1} \in \{1,2\}$,
and 
$\circledcirc_{N,D}$ 
to  be the set of $N+1$ tuples over the alphabet $\{1,\dots, D+N+1, ``\south"\}$. 
For reasons that will become apparent in the proof of \cref{lem:PermTech} below,  the set  $\circledcirc_{N,D}$ is called the \emph{compass}, and its elements \emph{compass elements}. 
\end{definition}

Recall the notion of an unlabelled tree (Definition~\ref{defn:UnlabelledTree}).\begin{definition}\label{defn:addStick}
     Let $\overline T_0,\overline T_1$ be unlabelled rooted trees. Define $\overline T_0\addRoot \overline T_1$ if the degree of the root of 
           $\overline T_1$ is 1, and   $\overline T_0$ is the graph minor of  $\overline T_1$ obtained by contracting the edge $\{\epsilon, 1\}$.
           See Figure~\ref{fig:newTree}.
 It follows that $\nu\in V(\overline T_0)$ if and only if $1\nu\in V(\overline T_1)$.

           \begin{figure}[h]

\begin{tabular}{ccc}
 \begin{tikzpicture}[scale=.9]

    \draw (-2,2.5) --  (0,1.2) --  (-0.7,2.5);
      \draw (.7,2.5) --  (0,1.2) --  (2,2.5);
       \draw [dashed] (-2,2.5) --  (-2,3.7);
     
            \draw [dashed]   (-.7,2.5)--(-1.4,3.7);
            \draw [dashed]  (-.7,2.5)--(0,3.7); 
                     \draw [dashed] (2,2.5) --  (2,3.7);

 \draw (0,0.3) node {\phantom{$\bullet$}};
  \draw  (-.7,2.5) node {$\bullet$};
    \draw  (.7,2.5) node {$\bullet$};
      \draw  (-2,2.5) node {$\bullet$};
        \draw  (2,2.5) node {$\bullet$};

\draw (0,1.2) node {$\bullet$};
\draw (-.3,1.2) node {\color{cyan}\tiny$\varepsilon$};

               \draw (-2.2,2.3) node {\color{cyan}\tiny$1$};  
          \draw (-.9,2.3) node {\color{cyan}\tiny$2$};
        \draw (.9,2.3) node {\color{cyan}\tiny$3$};
        \draw (2.2,2.3) node {\color{cyan}\tiny$4$};

\end{tikzpicture}
& 
 \begin{tikzpicture}[scale=.9]
   \draw[thick, magenta]  (0,0.3) --  (0,1.2);
    \draw (-2,2.5) --  (0,1.2) --  (-0.7,2.5);
      \draw (.7,2.5) --  (0,1.2) --  (2,2.5);
       \draw [dashed] (-2,2.5) --  (-2,3.7);
           
                         \draw [dashed]   (-.7,2.5)--(-1.4,3.7);
            \draw [dashed]  (-.7,2.5)--(0,3.7); 
                     \draw [dashed] (2,2.5) --  (2,3.7);

 \draw (0,0.3) node {$\color{magenta}\bullet$};
  \draw  (-.7,2.5) node {$\bullet$};
    \draw  (.7,2.5) node {$\bullet$};
      \draw  (-2,2.5) node {$\bullet$};
        \draw  (2,2.5) node {$\bullet$};

\draw (-.2,0.3) node {\color{magenta}\tiny$\varepsilon$};
\draw (0,1.2) node {$\bullet$};
\draw (-.3,1.2) node {\color{magenta}\tiny$1$};

               \draw (-2.2,2.3) node {\tiny\color{magenta}$1$\color{cyan}$1$};  
          \draw (-.9,2.3) node {\tiny\color{magenta}$1$\color{cyan}$2$};
        \draw (.9,2.3) node {\tiny\color{magenta}$1$\color{cyan}$3$};
        \draw (2.2,2.3) node {\tiny\color{magenta}$1$\color{cyan}$4$};

\end{tikzpicture}

\\  $\overline T_0$ && $\overline T_1$
\end{tabular}

\caption{Unlabelled trees  $\overline T_0\addRoot \overline T_1$
as in \cref{defn:addStick}. Recall  that the small digits drawn at each vertex are addresses and not labels.
\label{fig:newTree}}
\end{figure}
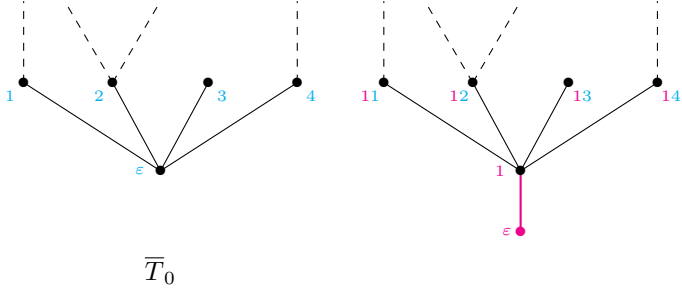

         \end{definition}

 We are now ready to prove the key lemma.

\begin{lemma}\label{lem:PermTech}
Let $k,N\in\N_+$, $\sigma\in S_N$, and $L_\sigma=\{w_{\sigma(1)}\cdots w_{\sigma(N)}\mid w_1\cdots w_N\in L\}$. If $L$ is $k$-\MCF\ then 
 $L_\sigma$ is  $(\constPerm)$-\MCF. 
\end{lemma}
\begin{proof}
Assume $D\in \N_+$ and $\mathcal A_N = (Q[N],\CCC[N], \SigHash{N}, q_0,  \delta_N, Q_f)$ is the $k$-restricted TSA constructed in  Lemma~\ref{lem:InsertHashes} of degree $D+N+1$ accepting $L_N \subseteq \SigHash{N}^* $ with states 
 $Q[N]=Q\sqcup Q_{\text{spec}}$
and  tree-label alphabet $C[N]=C\sqcup C_{\text{spec}}$ 
where \[Q_{\text{spec}}= \{ \llbracket i,q,q',c\rrbracket_j\mid i\in[1,N+1], q,q'\in Q, c\in C,j\in\{1,2\}\}\] 
and
 \[C_{\text{spec}}=\{\llbracket i,q,q',c\rrbracket \mid  i\in[1,N+1],   q,q'\in Q,c\in C\},\] and let $\delta_N$ (resp. $\delta$) be the transitions of $\cA_N$ (resp. of $\cA$)  as in Lemma~\ref{lem:InsertHashes}. Let   $H_{N,C},H_{N,@}$ and $\circledcirc_{N,D}$ be as in Definition~\ref{defn:setsHCirc}.

We aim to construct a $(\constPerm)$-TSA $\APerm$  with alphabet $\SigHash{N}$, $\deg(\APerm)=\deg(\cA_N)=D+N+1$,  states $\mathfrak Q$  and tree-label alphabet $\mathfrak C$  listed below,
 which will accept (ignoring parentheses)
  \begin{equation}\label{Eq:acceptedWords}
\left(\#_{\sigma(1)}w_{\sigma(1)}\#_{\sigma(1)+1}\right)
\left(\#_{\sigma(2)}w_{\sigma(2)}\#_{\sigma(2)+1}\right)
\cdots
\left(\#_{\sigma(N)}w_{\sigma(N)}\#_{\sigma(N)+1} \right)
\end{equation} if and only if $w_1\cdots w_N\in L$. Applying an erasing homomorphism (Proposition~\ref{prop:closure_props}
\eqref{item:3}) 
to remove all $\#_i$ symbols will then give the result.

The states $\mathfrak Q$ consist of the following disjoint sets:
\bi\item a start state $\aleph_{\text{start}}$ and four additional states $\aleph_{\text{start},j}$ for $j\in[1,4]$
\item  $\{0,1\}^{N+1}$, the set of   all binary strings of length $N+1$
\item $ \{ \aleph_i, \beth_i\mid i\in[1,N]\}$
\item 
$Q\times [1,N]\times[1,7]$
\item $Q_{\text{spec}}'=
\{ \llbracket i,q,q',c\rrbracket_j\mid i\in[1,N+1], j\in[1,4], q,q'\in Q, c\in C\}\supset Q_{\text{spec}}$
\item $Q_{\text{Phase3}}$ which is a finite set of 
 additional states, 
\ei
and  the  tree-label alphabet is \[\mathfrak C=
  C_{\text{spec}}\sqcup \left(\left(H_{N,C}\sqcup H_{N,@}\right)\times \circledcirc_{N,D}\times (C_@\sqcup \{-\})\right)\sqcup \{\Box\}\]   
where $-,\Box\notin C_@$ and distinct symbols. Note that $@\notin \mathfrak C$ as per \cref{defn:TreeWithLabels}.

The TSA $\APerm$ will be designed such that the following property holds. 
 \begin{propx}\label{SENTENCE-REF}
A pair  of labeled trees $(T_0,T_1)$ \emph{satisfy Property~\ref{SENTENCE-REF}} if following properties hold.
\be

\item  $\overline T_0\addRoot \overline T_1$ where  $\overline T_0,\overline T_1$ are the  unlabelled trees of $T_0,T_1$ respectively

\item 
 $\#_1w_1\#_2w_2\cdots \#_Nw_N\#_{N+1}$ is accepted by a run of $\cA_N$
 with final tree $T_0$ if and only if
\[\left(\#_{\sigma(1)}w_{\sigma(1)}\#_{\sigma(1)+1}\right)
\left(\#_{\sigma(2)}w_{\sigma(2)}\#_{\sigma(2)+1}\right)
\cdots
\left(\#_{\sigma(N)}w_{\sigma(N)}\#_{\sigma(N)+1} \right)\]
is accepted by a run of $\APerm$  with final tree $T_1$ 

\item $\nu$ is an $i$-th special vertex of $T_0$ labelled by $\llbracket i,q,q',c\rrbracket \in C_{\text{spec}}$  if and only if $T_1$ has  vertex $1\nu$ labelled 
  $\llbracket i,q,q',c\rrbracket$ if and only if $\nu=\nu'(D+i)$ for some $\nu'\in \N^*$. 

\ee

    \end{propx}

We call a vertex of $T_1$ labeled by $\llbracket i,q,q',c\rrbracket \in C_{\text{spec}}$ 
an \emph{$i$-th $\sigma$-special vertex} of  $T_1$.

The TSA $\APerm$  will operate in three 
 \emph{phases}: Phase One to non-deterministically construct a labeled tree, Phase Two  to simulate a run  of $\cA_N$ using the tree constructed in the first phase, and Phase Three to verify the non-deterministic choices in the first phase were valid.
 Phase Two will comprise  $N$ subroutines which are applied in a specific order.
 Each phase and subroutine  uses a distinct subset of states (aside from the first and last states to transition between them) to enforce that $\APerm$  applies each phase and subroutine in the intended order.
 We therefore present the transitions of $\APerm$ in the form of an algorithm.

\subsection*{Phase One} 
Phase One  consists of a set of transitions which can be used to non-deterministically construct a tree stack $(T_1,\varepsilon)$ with root of degree 1 and out-degree of every other vertex at most $D+N+1$. It does not read any input letters. 
It  uses states $\aleph_{\text{start}}$, $\aleph_{\text{start},j}$ for $j\in[1,4]$, 
$ \{0,1\}^{N+1}$ and $\alephEND$,
 and the following five types of transitions.  The state $\alephEND$ will be the final state for this phase.

\subsubsection*{1. (Push a copy of the root and first and last special vertices)}
For all  $\mathfrak h\in H_{N,@}$ and $\mathfrak c\in \{(1, x_2, \dots, x_N, D+N+1) \mid x_i\in  [1,  D+N], i\in [2,N]\}\subset \circledcirc_{N,D}$   we have a transition
  \begin{equation}\label{transition1-1} 
 (\aleph_{\text{start}}, \equals(@), \varepsilon, \push_1((\mathfrak h,\mathfrak c,@)),  \aleph_{\text{start},1}).
\end{equation} 
 
 For all $q',q''\in Q$, 
$\mathfrak h\in H_{N,@}$, 
$\mathfrak c\in \{(1, x_2, \dots, x_N, D+N+1) \mid x_i\in  [1,  D+N+1], i\in [2,N]\}\subset \circledcirc_{N,D}$,
 we have   transitions
  \begin{equation}\label{transition1-2a} 
    \begin{split} 
&(\aleph_{\text{start},1},\varepsilon, \equals((\mathfrak h,\mathfrak c, @)), \push_{D+1}(\llbracket 1 ,q_0, q',@\rrbracket), \aleph_{\text{start},2})\\
&(\aleph_{\text{start},2},\varepsilon, \equals(\llbracket 1 ,q_0,q',@\rrbracket), \down, \aleph_{\text{start},3})\\
&(\aleph_{\text{start},3},\varepsilon, \equals((\mathfrak h,\mathfrak c, @)), \push_{D+N+1 }(\llbracket N+1 , q'',q_f,@\rrbracket), \aleph_{\text{start},4})\\
&(\aleph_{\text{start},4},\varepsilon, \equals(\llbracket N+1 ,q'',q_f,@\rrbracket), \down, 10^{N-1}1)\end{split}\end{equation}

See \cref{fig:buildTEMP2} for an example of a tree $T_1$ constructed by a sequence of  transitions of type 1.

\subsubsection*{2. (Push a copy of a remaining  special vertex)}

For all $i \in [2,N]$, 
$\mathfrak h\in H_{N,C}\sqcup H_{N,@}$, 
$\mathfrak c\in \circledcirc_{N,D}$ where $\mathfrak c$   has $D+i$ as its  $i$-th   coordinate,  
$c_0\in\{@,-\}$,
  $\llbracket i ,q,q',c\rrbracket\in C_{\text{spec}}$,
$\aleph'=b_1\cdots b_{N+1}$ with
  $a_i =0$, $b_i =1$ and $a_j=b_j\in \{0,1\}$ for $j\neq i$,
 we have a transition
   \begin{equation}\label{transition1-2b} 
   (\aleph,\varepsilon, \equals((\mathfrak h,\mathfrak c, c_0)), \push_{D+i }(\llbracket i ,q,q',c\rrbracket), \aleph').\end{equation}

\subsubsection*{3. (Push  copy of a non-root non-special vertex)}
For all 
$\ell \in [1,D+N+1]$, $\mathfrak h\in H_{N,C}\sqcup H_{N,@}$,  $\mathfrak h'\in H_{N,C}$, 
  $c_0\in\{@,-\}$, $\aleph \in \{0,1\}^{N+1}$ and 
$\mathfrak c,\mathfrak c' \in\circledcirc_{N,D}$  satisfying the following conditions for all $i\in [1,N+1]$:
\bi
\item if $\mathfrak c$ has $\ell$ as its $i$-th coordinate, then the $i$-th coordinate of $\mathfrak c'$ is $j\in[1,D+N+1]$

\item if $\mathfrak c$ has ``south'' or $j\neq \ell$ as its $i$-th coordinate, the $i$-th coordinate of $\mathfrak c'$ is  ``south'',
\ei
then
we have a transition
  \begin{equation}\label{transition1-3} 
(\aleph,\varepsilon, \equals((\mathfrak h,\mathfrak c, c_0)), \push_{\ell}(\mathfrak h',\mathfrak c', - ), \aleph).\end{equation}

\subsubsection*{4. (Move down)} For all $\aleph\in \{0,1\}^{N+1}$ and
 $d\in \left(H_{N,C}\times \circledcirc_{N,D}\times \{-\}\right)\sqcup\left( \{0\}\times C_{\text{spec}}\right)$,
we have a transition
  \begin{equation}\label{transition1-4} 
 (\aleph,\varepsilon, \equals(d), \down, \aleph).\end{equation}

\subsubsection*{5. (Finish Phase One)} For all $\mathfrak h\in H_{N,@},\mathfrak c\in \circledcirc_{N,D}$, we have a transition
   \begin{equation}\label{transition1-5} 
   (1^{N+1},\varepsilon, \equals((\mathfrak h,\mathfrak c,@)), \id, \alephEND
).\end{equation}.

 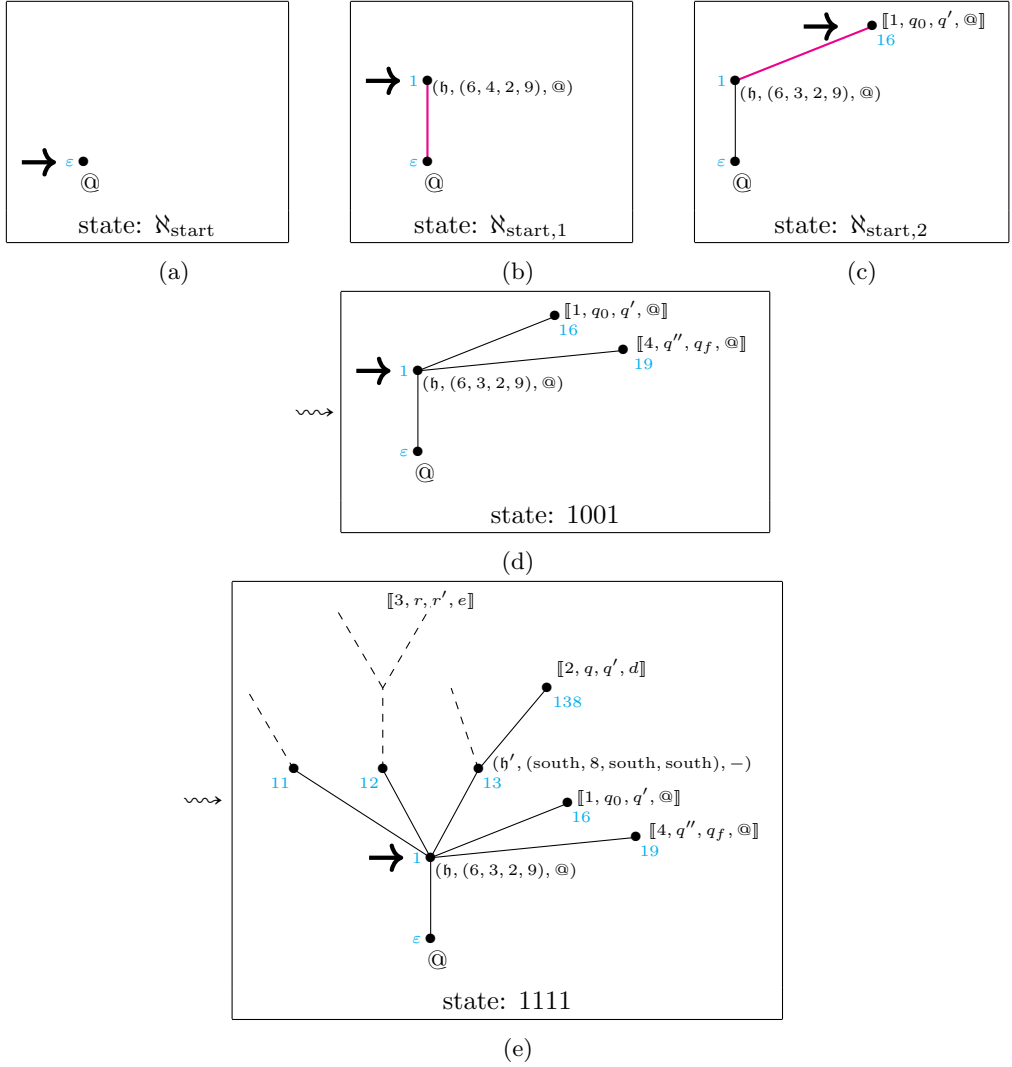
\begin{figure}[p]

\begin{subfigure}{0.3\textwidth}
\begin{tabular}{|c|}

\hline
 \begin{tikzpicture}[scale=.9]
 \draw (0.1,-0.3) node {$@$};
 \draw (0,0) node {$\bullet$};
\draw (-.2,0) node {\color{cyan}\tiny$\varepsilon$};

\phantom{    
               \draw (1.9,2.1) node{\tiny $\llbracket 1,q_0,q',@\rrbracket$};
           }
\draw[decorate,ultra thick,->] (-0.9,0) -- (-0.4,0);
\end{tikzpicture}

\\
state: $\aleph_{\text{start}}$\\
\hline
\end{tabular}
\caption{ 
\label{fig:buildA}}
\end{subfigure}
\begin{subfigure}{0.3\textwidth}
\begin{tabular}{|c|}
\hline
 \begin{tikzpicture}[scale=.9]
 
 \phantom{     
                  \draw (1.9,2.1) node{\tiny $\llbracket 1,q_0,q',@\rrbracket$};
           }         
           
  \draw[thick,magenta]  (0,0) --  (0,1.2);
  \draw (0.1,-0.3) node {$@$};
 \draw (0,0) node {$\bullet$};
\draw (-.2,0) node {\color{cyan}\tiny$\varepsilon$};
\draw (0,1.2) node {$\bullet$};
\draw (-.2,1.2) node {\color{cyan}\tiny$1$};
\draw (1.1,1.1) node {\tiny $(\mathfrak h,(6,4,2,9),@)$};

\draw[decorate,ultra thick,->] (-0.9,1.2) -- (-0.4,1.2);
\end{tikzpicture}
\\
state: 
$\aleph_{\text{start},1}$\\
\hline
\end{tabular}
\caption{ 
\label{fig:buildB}}
\end{subfigure}
\begin{subfigure}{0.3\textwidth}
\begin{tabular}{|c|}
\hline
 \begin{tikzpicture}[scale=.9]
 
    \draw [thick,magenta]   (0,1.2) --  (2,2);
  \draw  (0,0) --  (0,1.2);

  \draw (0.1,-0.3) node {$@$};
  
 \draw (0,0) node {$\bullet$};

\draw (-.2,0) node {\color{cyan}\tiny$\varepsilon$};
\draw (0,1.2) node {$\bullet$};
\draw (-.2,1.2) node {\color{cyan}\tiny$1$};
\draw (1.1,1.0) node {\tiny $(\mathfrak h,(6,3,2,9),@)$};

            \draw  (2,2) node {$\bullet$};
           
           \draw (2.9,2.1) node{\tiny $\llbracket 1,q_0,q',@\rrbracket$};
        \draw (2.2,1.8) node {\color{cyan}\tiny$16$};        
        \draw[decorate,ultra thick,->] (1,2) -- (1.5,2);
        \end{tikzpicture}

\\state:
 $\aleph_{\text{start},2}$\\
\hline
\end{tabular}
\caption{
\label{fig:buildTEMPAA}}
\end{subfigure}

\begin{subfigure}{0.4\textwidth}
$\rightsquigarrow$
\begin{tabular}{|c|}
\hline
 \begin{tikzpicture}[scale=.9]
  \draw  (0,0) --  (0,1.2);

  \draw (0.1,-0.3) node {$@$};
  
 \draw (0,0) node {$\bullet$};

\draw (-.2,0) node {\color{cyan}\tiny$\varepsilon$};
\draw (0,1.2) node {$\bullet$};
\draw (-.2,1.2) node {\color{cyan}\tiny$1$};
\draw (1.1,1.0) node {\tiny $(\mathfrak h,(6,3,2,9),@)$};

      \draw  (0,1.2) --  (3,1.5);
       \draw  (3,1.5) node {$\bullet$};
        \draw (4,1.6) node {\tiny $\llbracket 4,q'',q_f,@\rrbracket$};
        \draw (3.3,1.3) node {\color{cyan}\tiny$19$};       
            \draw  (2,2) node {$\bullet$};
              \draw   (0,1.2) --  (2,2);
           \draw (2.9,2.1) node{\tiny $\llbracket 1,q_0,q',@\rrbracket$};
        \draw (2.2,1.8) node {\color{cyan}\tiny$16$};

        \draw[decorate,ultra thick,->] (-0.9,1.2) -- (-0.4,1.2);
        \end{tikzpicture}

\\state:
 $1001$\\
\hline
\end{tabular}
\caption{ 
\label{fig:buildTEMP2}}
\end{subfigure}

\begin{subfigure}{0.6\textwidth}
$\rightsquigarrow$
\begin{tabular}{|c|}
\hline
 \begin{tikzpicture}[scale=.9]
  \draw  (0,0) --  (0,1.2);
    \draw (-2,2.5) --  (0,1.2) --  (-0.7,2.5);
      \draw (.7,2.5) --  (0,1.2); 
       \draw [dashed] (-2,2.5) --  (-2.7,3.7);
              \draw [dashed] (-.7,2.5) --  (-.7,3.7)--(0,4.9); 
               \draw [dashed]   (-.7,3.7)--(-1.4,4.9);
                     \draw (0.7,2.5) --  (1.7,3.7);
                    \draw [dashed]   (.7,2.5) --  (0.3,3.7);

  \draw (0.1,-0.3) node {$@$};
  
 \draw (0,0) node {$\bullet$};
  \draw  (-.7,2.5) node {$\bullet$};
    \draw  (.7,2.5) node {$\bullet$};
      \draw  (-2,2.5) node {$\bullet$};
  
                \draw  (1.7,3.7) node {$\bullet$};

\draw (-.2,0) node {\color{cyan}\tiny$\varepsilon$};
\draw (0,1.2) node {$\bullet$};
\draw (-.2,1.2) node {\color{cyan}\tiny$1$};

  \draw  (0,5) node {\tiny $\llbracket 3,r,r',e\rrbracket$};
        \draw  (2.5,4.0) node {\tiny $\llbracket 2,q,q',d\rrbracket$};
        
        \draw (2.8,2.6) node {\tiny $(\mathfrak h',(\text{south},8, \text{south},\text{south}),-)$};
         \draw (-2.2,2.3) node {\color{cyan}\tiny$11$};   \draw (-.9,2.3) node {\color{cyan}\tiny$12$};
        \draw (.9,2.3) node {\color{cyan}\tiny$13$};
 
            \draw (2,3.5) node {\color{cyan}\tiny$138$};

\draw[decorate,ultra thick,->] (-0.9,1.2) -- (-0.4,1.2);

\draw (1.1,1.0) node {\tiny $(\mathfrak h,(6,3,2,9),@)$};

        \draw (4,1.6) node {\tiny $\llbracket 4,q'',q_f,@\rrbracket$};
        \draw (3.2,1.3) node {\color{cyan}\tiny$19$};        

        \draw  (3,1.5) node {$\bullet$};
           \draw  (0,1.2) --  (3,1.5);
            \draw  (2,2) node {$\bullet$};
              \draw   (0,1.2) --  (2,2);
           \draw (2.9,2.1) node{\tiny $\llbracket 1,q_0,q',@\rrbracket$};
        \draw (2.2,1.8) node {\color{cyan}\tiny$16$};

\end{tikzpicture}
\\state: $1111$\\
\hline
\end{tabular}

\caption{ 
\label{fig:buildE}}
\end{subfigure}

\caption{
An example of transitions in Phase One of Lemma~\ref{lem:PermTech} with $N=3$ and $D=5$. Figures~\ref{fig:buildA}--\ref{fig:buildTEMP2}
show the effect of the 
transitions of type (1)~(push a copy of the root and first and last special node);
\cref{fig:buildE} shows the tree stack  obtained after further applications of transitions of type (2), (3), and (4), with the pointer at address $1$. 
\cref{fig:buildE} may (possibly) be followed by further transitions of type (3) and (4), and then a transition of type (5)~(finish Phase One) moving to state $ \alephEND$, since all three special vertices have been constructed (as indicated by the state $1111$).
\\
 The compass entry $({6},{3},{2},9)$ at vertex $1$ tells us that the \nth{1}  $\sigma$-special vertex (labeled $\llbracket {1},q_0,p',@\rrbracket$) can be reached from this vertex by traveling to vertex $1{6}$ since $D+1=5+1=6$,  the \nth{2} $\sigma$-special vertex   can be reached by traveling to vertex $1{4}$, the \nth{3} 
  can be reached by traveling to vertex $1{2}$, and the 
  \nth{4} 
    can be reached by traveling to vertex $1{9}$ since $D+N+1=5+3+1=9$.
 \\
  The compass entry $(\text{south},8, \text{south},\text{south})$ at vertex address $13$ tells us that the \nth{1} $\sigma$-special vertex 
   can be reached from this vertex by traveling down, 
   the \nth{2}  by traveling up to vertex $1{3}{8}$ (since $8=D+2$), and the \nth{3}
   by traveling down.
\label{fig:phaseOneA}}
\end{figure}

See Figure~\ref{fig:buildE} for an example of a tree $T_1$ constructed by a sequence of  transitions up to type 4.

These transitions 
ensure that 
a run of Phase One which starts in state $\aleph_{\text{start}}$ and ends in state $\alephEND$ 
must 
 operate as follows. Starting with the tree stack initiated at $(\{(\varepsilon,@)\},\varepsilon)$ and the TSA in state $\aleph_{\text{start}}$, the only possible transition is type (1), moving through states $\aleph_{\text{start},1}$ to  $\aleph_{\text{start},4}$ building a tree as in 
 \cref{fig:buildTEMP2}, ending in state $10^{N-1}1$ to indicate the first and last $\sigma$-special vertices have been added.
 After this $\APerm$  can non-deterministically perform transitions of type (2) and (3) to add new vertices (note by definition of a TSA, it cannot perform  $\push_\ell$ if the $\ell$-th child is already present) and (4) to move down to build new branches (not moving below the vertex with address 1, or up to any $\sigma$-special vertex already added).
 By construction, a state $\aleph \in \{0,1\}^{N+1}$  will have $1$ as its $i$-th digit if and only if an $i$-th $\sigma$-special vertex 
 has been added to the tree. Once a state has $1$ in its $i$-th coordinate, all subsequent states reached will have this property.    (Note there is no change to the state $\aleph$ for transitions of type (3) and (4) since no special vertex is added by these.) It follows that the tree constructed will have,  for each $i\in[1,N+1]$,  exactly one $i$-th $\sigma$-special  vertex at a node with address $\nu'(D+i)$ for some $\nu'\in N^*$.

 The compass element at each vertex
   encodes the relative locations of each $\sigma$-special vertex as follows:  if an $i$-th $\sigma$-special  vertex 
 is below the current vertex $\nu$, the compass element at $\nu$ has ``south'' as its $i$-th coordinate; 
 and if an $i$-th $\sigma$-special  vertex is located at or  above $\nu \ell$,  the compass element at $\nu$ has ``$\ell$'' as its $i$-th coordinate. Recall that $\sigma$-special vertices do not have a compass element (by construction, the location of $\sigma$-special vertices other than the current one is always ``south''). For an example, see vertices at addresses $1$ and $13$ in Figure~\ref{fig:phaseOneA}. 
 
  $\APerm$ can move to state $\alephEND$ and finish using a transition of type (5) if and only if   $\APerm$ is in state $1^{N+1}$ if and only if a tree stack $(T_1,1)$ has been built with exactly one  $i$-th $\sigma$-special label vertex
   for each $i\in[1,N+1]$ address 1 labeled by $(\mathfrak h,\mathfrak c, @)$ with $\mathfrak h\in H_{@,C}$, and  all other non-root vertices labeled by $(\mathfrak h,\mathfrak c, -)$ with $\mathfrak h\in H_{N,C}$.

 Since the only transitions used for  Phase One are $\push_j$, $\down$ and $\id$, each non-root vertex is visited from below exactly {once} in this phase.  
 For the remaining phases, $\APerm$ will not perform any more $\push_j$ transitions (so we will have $\deg(\APerm)=D+N+1$ by construction).  After Phase One is completed, the labels of the  $\sigma$-special vertices will not be altered by any  transition in the next two phases, and no state except the final $\aleph_{\sigma(1)}$ will be used again.

\subsection*{Subroutine A$[i]$}
We now describe,  for  $i\in [1,N]$,  a set of transitions  which can be employed as a subroutine to be used by Phase Two below.

For $i\in [1,N]$, the transitions of  subroutine A$[i]$ use the states  $\aleph_i$, $\beth_{i}$,  $
Q\times \{i\}\times [1,7]$  and $Q_{\text{spec}}'$.

Starting in state  $ \aleph_{i}$ with the pointer at address $1$, the following transitions enable us to  ``simulate''  a run of $\cA_N$ as it reads  the factor  $\#_{i}w_{i}\#_{i+1}$ of the word in \cref{Eq:acceptedWords}
and constructs/moves around a labelled tree $T_0$, but instead using the labeled tree $T_1$ built by Phase One,
with $\APerm$ finishing in state  $ \beth_{i}$
with the pointer back at address $1$.

\subsubsection*{1. (Locate $i$-th $\sigma$-special vertex)}  
For each $\ell\in[1,D+N+1]$, 
$\mathfrak c\in \circledcirc_{N,D}$ such that $\mathfrak c$ has $\ell$ in its $i$-th coordinate, 
$\mathfrak h\in H_{N,@}\sqcup H_{N,C}$,  
$c\in C_@\sqcup\{-\}$, 
we have the transition
\begin{equation}\label{transitionAi-1} (\aleph_i, \varepsilon, \equals((\mathfrak h,\mathfrak c, c)), \up_\ell, \aleph_i),\end{equation}
and for each 
 $\llbracket i,q,q',c\rrbracket\in C_{\text{spec}}$ we have the transition
\begin{equation}\label{transitionAi-2}
   (\aleph_i,\varepsilon, \equals( \llbracket i,q,q',c\rrbracket), \id, \llbracket i,q,q',c\rrbracket_2).\end{equation}

Following these transitions, the pointer moves up (following the $i$-th coordinate of the compass element of the vertex labels) until the pointer is at (the unique) $i$-th $\sigma$-special vertex of $T_1$, and $\APerm$ is in the state $\llbracket i,q,q',c\rrbracket_2\in Q_{\text{spec}}'$.

\subsubsection*{2. (Simulate $\mathcal A_N$ to read  $\#_i$)}  By construction (\cref{lem:InsertHashes}), 
 if $\cA_N$ is in state $ \llbracket i,q,q',c\rrbracket_2$, the pointer must be pointing to a vertex of $T_0$ labeled  $ \llbracket i,q,q',c\rrbracket$ which has address $\nu(D+i)$ for some $\nu\in \N^*$ such that $\nu$ has label $c$,
and there are unique transitions  $\tau_1(i,c),\tau_2(i,c), \tau_3(i,c)$ of $\cA_N$ to and from this state.

Thus, for every 
\bi\item transition
  $ \tau_3 (i,c) =(\llbracket i,q,q',c\rrbracket_2,\#_i,\equals(\llbracket i,q,q',c\rrbracket), \down, q')\in\delta_N$
\item 
$\mathfrak c\in \circledcirc_{N,D}$ with $i$ in its $(D+i)$-th coordinate
\item $d\in C_@, e\in C_@\cup\{-\}$
\item 
  $\mathfrak h\in H_{N,C}\sqcup H_{N,@}$ such that if $i>1$, the $(i-1)$-th and $i$-th columns are  $\begin{bmatrix} n_{i-1} & 0\\c & d\end{bmatrix}$ for some $n_{i-1}\in \{0,1,2\}$ (or for $i=1$,  the \nth{1} column is defined to be $ \begin{bmatrix}  1 \\@\end{bmatrix}$),
  \ei
we have transitions of $\APerm$:
\begin{equation}\label{transitionAi-3}
\begin{split}
&(\llbracket i,q,q',c\rrbracket_2,\#_i,\equals(\llbracket i,q,q',c\rrbracket), \down, \llbracket i,q,q',c\rrbracket_3)\\
&( \llbracket i,q,q',c\rrbracket_3, \epsilon,\equals((\mathfrak h, \mathfrak c,e)), \set((\mathfrak h', \mathfrak c,c)),  (q',i,1))\\
\end{split}\end{equation}
where $\mathfrak h'$ is   obtained from  $\mathfrak h$ by replacing its $i$-th coordinate by $\begin{bmatrix}1\\d\end{bmatrix}$.

 Note that for $i>1$, the second  transition overwrites the label $(\mathfrak h, \mathfrak c,e)$ of the vertex below the $i$-th $\sigma$-special label by $(\mathfrak h'', \mathfrak c,c)$ (and similarly for the third transition). This is because the $(i-1)$-th column of $\mathfrak h$ says that  in the tree 
  $T_0$ the corresponding vertex had the label $c$ after the run had finished reading $w_{i-1}\#_i$, 
   which will be verified when Subroutine $A(i-1)$ is run (which may be \emph{after} Subroutine $A(i)$ is run during Phase Two). The label $e$ is disregarded in these transitions.
  
After performing  these two transitions  $\APerm$ is in state $(q',i,1)$ with the pointer at  the vertex below the $i$-th $\sigma$-special node in $T_1$ with $c$ in the third coordinate of the vertex label, whereas the corresponding run of $\cA_N$ at this point of reading $\#_i$ has $\cA_N$ in state $q'$ with the pointer at  the vertex below the $i$-th special node in $T_0$ with $c$ as the vertex label.

\subsubsection*{2. (Simulate $\mathcal A_N$ to read  $w_i$)} In this step $\APerm$ will perform modified versions of the transitions of $\cA_N$, reading the third coordinate of the vertex label as if it were the label of the corresponding vertex in $T_0$.

 For each transition $\s=(q,x,p,f,q') \in \delta_N$  where $q\in Q, q'\in Q\cup\{\llbracket i+1,p,p',g\rrbracket_2\mid p,p'\in Q, g\in C_\varepsilon\}$ 
 and $p\in\{\true,\equals(e)\}$,
$x\in \Sigma_\varepsilon$, 
and each $(\mathfrak h,\mathfrak c,e)\in (H_{N,C}\sqcup H_{N,@})\times \circledcirc\times C_@$,
we have the following  transitions of $\APerm$ which will send $\APerm$ from state $(q,i,1)$ to  state $(q',i,1)$.
(Note that if $p=\equals (d)$ and the pointer is pointing to a vertex labelled $(\mathfrak h,\mathfrak c,e)$ with $e\neq d$, there is no transition.)

\be

\item If $f=\id$ we have the transition
\begin{equation}\label{transitionAi-4}
((q,i,1),x,p',\id, (q',i,1)).
\end{equation}  where $p'=\true $ if $p=\true$, and otherwise $p'=\equals((\mathfrak h,\mathfrak c,e))$.

 \item If $f=\set(c)$ we have transitions
\begin{equation}\label{transitionAi-5}
((q,i,1),x,\equals((\mathfrak h,\mathfrak c,d_1)),\set((\mathfrak h,\mathfrak c,c)), (q',i,1))\\
\end{equation}  
where $d_1=e$ if $p=\equals(e)$ and $d_1\in C$ if $p=\true$.

\item If $f=\down$, 
we have one set of transitions to adjust the source vertex label, and a second set of transitions to adjust the target vertex label. 

When performing a down transition, we first need to non-deterministically decide whether this is the last time the pointer is at the vertex it is leaving whilst reading $w_i\#_{i+1}$ or not.
We do this with the following transitions:
\begin{equation}\label{transitionAi-6}\begin{split}
((q,i,1),x,\equals((\mathfrak h,\mathfrak c,d_1)),\down, (q,i,2))\\
((q,i,1),x,\equals((\mathfrak h,\mathfrak c,d_1)),\set((\mathfrak h',\mathfrak c,e)), (q,i,3))\\
( (q,i,3),\epsilon,\true,\down, (q,i,2))\end{split}
\end{equation}  
where $d_1=e$ if $p=\equals(e)$ and $d_1\in C$ if $p=\true$, and
$\mathfrak h'$ is obtained from $\mathfrak h$ by replacing column $i$ of $\mathfrak h$ which is $\begin{bmatrix} 1\\ b_1\end{bmatrix}$
by $\begin{bmatrix} 2\\ b_1\end{bmatrix}$. Note that the entry in row 1 column $i$ must be 1 since the pointer has reaches this vertex whilst reading $\#_iw_i$.

For a vertex $\nu\in T_0$ to be below a vertex that has been visited by the pointer whilst $\cA_N$ is reading $w_i$, $\nu$ is non-special and must already have a label $d\in C_@$. Thus $1\nu\in T_1$ will have a label $(\mathfrak h',\mathfrak c',d)$. 
If this is the first time the pointer of $\cA_N$ has visited $\nu$ whilst reading $w_i$, then  $\mathfrak h'$ will have $0$ in its $i$-th column, and $i>1$ since the pointer cannot have visited previously, so  in order to know what the label of $\nu$ is after finishing reading $w_{i-1}$, we can look at the $(i-1)$-th column of $\mathfrak h'$ to find out.

We realise these considerations with the following transitions. For all $d\in C_@$ and all $\mathfrak h'\in H_{N,C}\sqcup H_{N,@}$ whose $i$-th column is  $\begin{bmatrix} 1\\ b\end{bmatrix}$ for $b\in C_@$, we have 
\begin{equation}\label{transitionAi-7}
(q,i,2),\epsilon,\equals((\mathfrak h',\mathfrak c',d)),\id, (q',i,1))
\end{equation}  
This means we do not modify the target vertex label since we have already visited it whilst reading $w_i$ and its label is ``correct''.

For all $d\in C_@$ and all  $\mathfrak h'\in H_{N,C}\sqcup H_{N,@}$ whose $(i-1)$-th and $i$-th columns are   $\begin{bmatrix} n_{i-1}&0\\ b_1 & b_2\end{bmatrix}$, we have 

\begin{equation}\label{transitionAi-8}
(q,i,2),\epsilon,\equals((\mathfrak h',\mathfrak c',d)),\set((\mathfrak h',\mathfrak c',b_1)),\down, (q',i,1))
\end{equation}  where $\mathfrak h''$ is obtained from $\mathfrak h'$ by replacing its $i$-th column by $\begin{bmatrix} 1\\ b_2\end{bmatrix}$.

\item If $f=\up_n$, we again need to decide whether this is the last time the pointer visits the source node whilst reading $w_i$, so we have 
transitions:
\begin{equation}\label{transitionAi-9}\begin{split}
((q,i,1),x,\equals((\mathfrak h,\mathfrak c,d_1)),\up_n, (q,i,4))\\
((q,i,1),x,\equals((\mathfrak h,\mathfrak c,d_1)),\set((\mathfrak h',\mathfrak c,e)), (q,i,5))\\
( (q,i,5),\epsilon,\true,\up_n, (q,i,4))\end{split}
\end{equation}  
where $d_1=e$ if $p=\equals(e)$ and $d\in C$ if $p=\true$, and 
 $\mathfrak h'$ is obtained from $\mathfrak h$ by replacing column $i$ of $\mathfrak h$ which is $\begin{bmatrix} 1\\ b_1\end{bmatrix}$ 
by $\begin{bmatrix} 2\\ b_1\end{bmatrix}$. 

Next we check column $i$ of the first coordinate of the vertex label at the vertex the pointer moves to. Note that we can only  reach a special vertex in $T_0$ with a $\push_n$ transition, so the label of the target vertex in $T_1$ here is 
$(\mathfrak h'',\mathfrak c'',b_1)$. 
For all $\mathfrak h''\in H_{N,C}$ such that column $i$  of $\mathfrak h''$ is $\begin{bmatrix} 1\\ b_2\end{bmatrix}$, and all $b_1\in C$ we have
 the following transitions.
\begin{equation}\label{transitionAi-10}
( (q,i,4),\epsilon,\equals((\mathfrak h'',\mathfrak c'',b_1)),\id, (q',i,1)).
\end{equation}

\item If $f=\push_n(g_1)$ for $g_1\in C\sqcup C_{\text{spec}}$,
we again need to decide whether this is the last time the pointer visits the source node whilst reading $w_i$, so we have 
transitions:
\begin{equation}\label{transitionAi-9}\begin{split}
((q,i,1),x,\equals((\mathfrak h,\mathfrak c,d_1)),\up_n, (q,i,6))\\
((q,i,1),x,\equals((\mathfrak h,\mathfrak c,d_1)),\set((\mathfrak h',\mathfrak c,e)), (q,i,7))\\
( (q,i,7),\epsilon,\true,\up_n, (q,i,6))\end{split}
\end{equation}  
where $d_1=e$ if $p=\equals(e)$ and $d_1\in C$ if $p=\true$, and 
 $\mathfrak h'$ is obtained from $\mathfrak h$ by replacing column $i$ of $\mathfrak h$ which is $\begin{bmatrix} 1\\ b_1\end{bmatrix}$ 
by $\begin{bmatrix} 2\\ b_1\end{bmatrix}$.

Next we must check that the target vertex has not been visited before whilst reading  $w_i$ or  any previous factor $w_j$ with $j<i$.
For all $g_1\in  C$ and all $(\mathfrak h',\mathfrak c',-)\in H_{N,C}$  and such that the $i$-th column of $\mathfrak h'$ is $\begin{bmatrix} 0\\g_1\end{bmatrix}$  and   if $i>1$, the $(i-1)$-th column of $\mathfrak h'$ is $\begin{bmatrix} \%\\\%\end{bmatrix}$,
we have transitions:
\begin{equation}\label{transitionAi-4}
((q,i,6),\epsilon,\equals((\mathfrak h',\mathfrak c',-)), \set((\mathfrak h'',\mathfrak c',g_1)),(q',i,1))
\end{equation} where $\mathfrak h''$ is obtained from $\mathfrak h'$ by replacing its $i$-th column by $\begin{bmatrix} 1\\g_2\end{bmatrix}$.

In addition we have these transitions for the case that the target vertex is $\sigma$-special, which is the case when $g_1=(\llbracket i+1,p,p',g\rrbracket)$.
For all $p,p''\in Q, g\in C_\varepsilon$
we have transitions:
\begin{equation}\label{transitionAi-4}
((q,i,6),\epsilon,\equals(\llbracket i+1,p,p',g\rrbracket), \id,\llbracket i+1,p,p',g\rrbracket_2)
\end{equation}
which, if followed, means that 
$\APerm$ has  successfully simulated $\cA_N$ reading $\#_iw_{i}$ moving from the $i$-th to the $(i+1)$-th special vertices of $T_0$ (resp. $\sigma$-special vertices of $T_1$).

 \ee

\subsubsection*{3. (End of reading $w_i\#_{i+1}$)}  
By construction (\cref{lem:InsertHashes}), 
 if $\cA_N$ is in state $ \llbracket i+1,p,p',g\rrbracket_2$, 
the last transition it performed must have been 
  $ \tau_2 (i+1,g)= (\llbracket i,p,p',g\rrbracket_1,\epsilon, \equals(g),\push_{D+i+1}( \llbracket i+1,p,p',g\rrbracket), \llbracket i+1,p,p',g\rrbracket_2)$
  followed by $ \tau_3 (i+1,g)$ which reads $\#_{i+1}$.

Thus, for every \bi\item transition
  $ \tau_2 (i+1,g)= (\llbracket i+1,p,p',g\rrbracket_1,\epsilon, \equals(g),\push_{D+i+1}( \llbracket i+1,p,p',g\rrbracket), \llbracket i+1,p,p',g\rrbracket_2)$\item  $\mathfrak n\in H_{N,C}\times \circledcirc\times C_-$, \item $\mathfrak m\in H_{N,@}\times \circledcirc\times \{@\}$, \ei
we have the following transitions of $\APerm$:
\begin{equation}\label{transitionAi3-1}\begin{split}
&(\llbracket i+1,p,p',g\rrbracket_2, \#_{i+1}, \equals(\llbracket i+1,p,p',g\rrbracket),  \id,  \llbracket i+1,p,p',g\rrbracket_4)\\
&( \llbracket  i+1,p,p',g\rrbracket_4, \epsilon,  \equals(\mathfrak n),  \down,  \llbracket  i+1,p,p',g\rrbracket_4)\\
&( \llbracket  i+1,p,p',g\rrbracket_4, \epsilon,  \equals(\mathfrak m), \id, \beth_{i}),\end{split}\end{equation}
which move the pointer to the vertex at address $1$ and $\APerm$ into state $\beth_{i}$.

By construction, the transitions described forming Subroutine $A(i)$ are the only possible to be performed in the order described, and it follows that $\APerm$ reads $\#_iw_i\#_{i+1}$ starting from $\aleph_i$ and ending at $\beth_{i}$ with the pointer starting and ending at the vertex at position $1$.

\subsection*{Phase Two}

Phase Two applies the transitions from the above subroutine  as follows.  For $i\in [1,N-1]$,  $\APerm$ has the following transitions:
\begin{equation}\label{transitionPhase2}
 \tau[i]=
 (\beth_{\sigma(i)}, \epsilon, \true,\id, \aleph_{\sigma(i+1)}).\end{equation}

Since  Phase One ends with $\APerm$ in state $ \aleph_{\sigma(1)}$ with the pointer at address $1$, the only possible transitions from here are from Subroutine A$[\sigma(1)]$, which means $\APerm$  reads $\#_{\sigma(1)}w_{\sigma(1)}\#_{\sigma(1)+1}$, with the pointer returned to the address $1$, leaving $\APerm$ in state $ \beth_{\sigma(1)}$. 
From here, the only possible transition is $\tau[1]$, which puts $\APerm$ into state $\aleph_{\sigma(2)}$, then from here the only possible  transitions 
are from Subroutine A$[\sigma(2)]$, then $\tau[2]$, and so on until the entire word in \cref{Eq:acceptedWords} has been read, the pointer is at address $1$, and $\APerm$ is in state $ \beth_{\sigma(N)}$.

At the end of Phase Two, we have verified that a run of $\cA_N$ exists starting at the special vertex for $\#_1$ to the special vertex reading $\#_{N+1}$, 
so this final step completes the verification that Property~\ref{SENTENCE-REF} holds provided the following final check is performed.

\subsection*{Phase Three} This phase starts in state $\beth_{\sigma(N)}$.
We wish to scan the labelled tree $T_1$  produced at the end of after Phase Two using a depth-first search to check that the first row of the first coordinate $\mathfrak h\in H_{N,C}\sqcup H_{N,@}$ of the label of each non-root, non-special vertex all have  $n_{j}\in\{2,0\}$ 
 (with $0$ indicating that the pointer didn't visit this vertex when reading $w_{i}$),
to ensure that  the last label of every vertex before reading $\#_j$ matches the operation of $\mathcal A_N$ when first entering the same vertex after reading $\#_j$ for each $j\in[1,N]$. That is, if  some label has $\mathfrak h$ with $\begin{bmatrix}1\\c\end{bmatrix}$ as its $j$-th column, we have not checked that the label of this vertex is $c$ the last time the pointer was at this vertex, and so the assumption that it was $c$ in Subroutine $A(j+1)$ may be incorrect, so this Phase will not allow any such transitions. It also will check that the first $n_j$ after a $\%$ entry is $2$, to ensure that the vertex was constructed by a $\push_n$ transition.

Since the tree $T_1$ has out-degree at most $D+N+1$, 
it is clear that, using  some additional states $Q_{\text{Phase3}}$ and writing $\Box$ at each vertex after it has been checked, this depth-first search can be implemented by finitely many transitions, and we omit the details. (To comply with \cref{rmk:finish-root}, the last transition should return the pointer to the root of $T_1$.)

This  phase ensures that $\APerm$ using the tree $T_1$ guessed in Phase One has correctly simulated $\mathcal A_N$ constructing final tree $T_0$ in Phase Two, starting from the start state of $\mathcal A_N$ with the pointer at the root of $T_0$, ending at an accept state  of $\mathcal A_N$ with the pointer at the root of $T_0$, and each vertex has been entered and exited from above and below (with the correct vertex label exiting the last time before reading $\#_i$ for each $i\in[2,\dots, N+1]$) exactly as guessed in Phase One. This proves  completeness and soundness.

Let $T_1$ be the final  tree of a run of $\APerm$.
The vertex of $\overline{T_1}$ with address $1$ is visited from below exactly once during the entire process (during Phase One; after this the pointer never drops below $1$ until the end of Phase Three). 
Phase One visits every vertex $1\nu$ of the tree $\overline{T_1}$ from below exactly once.
For $i\in[1,N]$, Subroutine $A(i)$ 
moves the pointer up from address $1$ to an $i$-th $\sigma$-special vertex, then simulates $\cA_N$ until it moves up to an  $i$-th $\sigma$-special vertex,  and then at the end of Subroutine $A(i)$ the pointer goes down to the root. Thus in total for Phase Two, 
a  $j$-th $\sigma$-special vertex is visited from below exactly once for $j\in\{1,N+1\}$ and twice for $j\in [2,N]$.
A non-special vertex that is not the root of $T_1$ is visited once for each application of Subrountine $A(i)$, so $N$ times, plus $k$ in total as each subroutine simulates $\cA_N$ building $T_0$ where each vertex of $T_0$ is visited at most $k$ times in total by hypothesis. Thus for Phase Two 
 the number of times a vertex is visited from below in Phase Two is  most $(N+1)+k$ since $N\geq 1$ so $N+1+k\geq 2$.
 Phase Three performs a depth-first search of $T_1$ starting at the address labeled $1$, which means it visits each  vertex  $1\nu$ of $T_1$ from below exactly once.

Thus in total all phases visit each vertex of $T_1$ from below at most \[1+(k+N+1)+1=k+N+3\]
 times, which shows that $\APerm$ is $(\constPerm)$-restricted.
\end{proof}

\section{Conclusion and outlook}

We have shown that $C^N$ of a $k$-\MCF\ language is $(\constPerm)$-\MCF. We do not know if the bound is sharp; for example, for $k=1$ and $N=3$ we have shown $C^3$ of a context-free language is a $7$-\MCF, whereas the lower bound is 2  since context-free languages are not closed under $C^3$ by \cite{Brandst}.

Our result adds the class of \MCF\ languages to the collection of formal languages classes which are closed under the operator $C^N$.

\section*{Acknowledgements}

 The first author was supported 
 by The Leverhulme Trust, Research Project Grant RPG-2022-025. The second and forth author were supported by an 
Australian Research Council grant DP210100271. The second author received support from the London Mathematical Society
Visiting Speakers to the UK--Scheme 2.  

\bibliographystyle{plainurl}
\bibliography{refs}

\end{document}